%% file: main.tex
\documentclass[sigconf,10pt]{acmart}


\usepackage{amsmath}
\usepackage{amsthm}
\usepackage{proof}
\usepackage{graphicx}
\usepackage{xspace}
\usepackage{color} 
\usepackage[all]{xy}
\usepackage{tikz-cd}
\usepackage{hyperref}
\usepackage{mathtools}
\usepackage{macros}
\usepackage[shortlabels]{enumitem}
\usepackage{libertine}
\usepackage{tikz}
\usepackage{todonotes}

\begin{document}

\title{Rank-decreasing transductions}
 \author{Miko{\l}aj Boja\'nczyk and Pierre Ohlmann}
 \begin{abstract}
    We propose to study transformations on graphs, and more generally structures, by looking at how the cut-rank (as introduced by Oum) of subsets is affected when going from the  input structure to the output structure. We consider transformations in which the underlying sets are the same for both the input and output, and so the cut-ranks of subsets can be easily compared. The purpose of this paper is to give a characterisation of logically defined transductions that is expressed in purely structural terms, without referring to logic:  transformations which decrease the cut-rank, in the asymptotic sense, are exactly those that can be defined in monadic second-order logic. This characterisation assumes that the transduction has  inputs of bounded treewidth; we also show that the characterisation fails in the absence of any assumptions.
\end{abstract}

\maketitle

\input{intro}
\input{overview}

\input{ranks}
\input{tree-transductions}

\input{extensions}

\bibliographystyle{plain}
\bibliography{bib}

\newpage
\appendix
\input{appendix}

\input{linear-case}

\end{document}

%% file: intro.tex
\section{Introduction}
One of the central themes of logic in computer science is the  equivalence of recognisability (by automata, or algebras) and definability in logic, see~\cite{Thomas97} and~\cite{bojanczyk_recobook} for surveys. This equivalence was proved originally for strings by B\"uchi~\cite[Thms 1 and 2]{Buchi60}, Elgot~\cite[Thm 5.3]{Elgot61} and Trakhtenbrot~\cite[Thms 1 and 2]{trakthenbrot1958}. It has been extended to many other structures, such as infinite words~\cite[Thm 1]{Buchi62}, or finite~\cite[Thm 17]{thatcher1968generalized} and infinite trees~\cite[Thm 1.7]{rabinDecidabilitySecondOrderTheories1969}. The logic that appears in such theorems is usually some variant of monadic second-order logic \mso. 

In words and trees, there are usually natural automata models, and the challenges are mainly about dealing with infinity. In graphs,  the challenge is different -- it is not clear how a graph should be parsed by an automaton.
The study of recognisability and logic for graphs started in the  late 1980's, mainly with the work of  Courcelle.    The resulting theory, see the book~\cite{courcelleGraphStructureMonadic2012} for a reference, is technically difficult, and is tightly connected to structural graph theory, in particular to treewidth. 
 (One can  consider more general structures, such as ternary relations, or even hypergraphs, but graphs contain already many of the important difficulties. We will talk about more general structures later in the paper. ) As was the case for words or trees, a central theme in this theory is 
\begin{align*}
\text{definable in \mso}
\quad \Leftrightarrow \quad
\text{recognisable}.
\end{align*}
The $\Rightarrow$ implication is known as  Courcelle's Theorem~\cite[Thm 4.4]{courcelleMonadicSecondorderLogic1990}. This is a  highly influential result, especially due to   its application to linear time   \mso model checking on graphs of bounded treewidth, which  initiated the field of algorithmic meta-theorems~\cite{grohe2011methods}. The implication $\Leftarrow$ is known to fail in general, but has been proved for various kinds of tree-decomposable graphs, such as graphs of bounded treewidth~\cite[Theorem 2.10]{bojanczykDefinabilityEqualsRecognizability2016a} or graphs of bounded linear cliquewidth~\cite[Theorem 3.5]{linearcliquewidth2021}. It is conjectured that the $\Leftarrow$ implication holds for all tree-decomposable graphs~\cite[Conjecture 5.3]{DBLP:journals/corr/abs-2305-18039}.

The equivalence of definability and recognisability  explains  the importance of \mso logic, by describing it in an alternative way that does not refer to logic. This way we know that \mso is not just some fluke of the syntax.  In this paper, we initiate a similar apologetic project, except  that instead of languages we consider transductions. 

Transductions, as considered in this paper, are transformations that input graphs and output other graphs, in a nondeterministic fashion (one input may yield several outputs).  The importance of transductions, and especially those  that can be defined in \mso, was observed independently by  all early authors in this topic, namely by Arnborg et al.~\cite[Sec 4]{arnborg1991easy}, Engelfriet~\cite[Def 6]{engelfriet1991} and Courcelle~\cite[Def 2.2]{courcelle1991}. 
The goal of this paper is to prove that \mso transductions are the canonical choice of graph transformation, by characterising them in a way that does not use logic, and refers only to the structural properties of the input and output  graphs. Informally speaking, our main result says that a transduction is definable in \mso if and only if simple subsets of the input graph are mapped to simple subsets of the output graph.  

Let us explain this informal statement in more detail.

We begin with the notion of  simple subsets in a graph.  The general idea is that to each subset of vertices in a graph we can assign a number, called its rank, which quantifies the amount of communication between the subset and its complement.  This rank can be formalised in a number of different ways, including  a Myhill-Nerode style definition as used in clique-width, or an algebraic definition proposed by Oum~\cite[Sec 2]{oum2005rank} which computes the rank of a certain matrix.  Although the various formalisations of rank give  different numbers, these numbers are  asymptotically equivalent, which means that one is bounded if and only if the other is bounded. Such asymptotic equivalence is all we need in this paper. 

We now explain what we mean by mapping simple subsets to simple subsets. In order to relate subsets of the input graph with subsets of the output graph, we consider transductions in which the input and output graphs use the same vertex set. This allows us to identify subsets of the input graph with subsets of the output graph, as in the following picture: 
    \mypic{25}
(The usual notion of transduction allows to remove some vertices and copy some others, the restriction on having the same vertex set will not incur any real loss of generality, while simplifying notation.) 
We say that a transduction is  \emph{rank-decreasing} if for every $k$ there is some $\ell$ such that if a subset of an input graph has rank at most $k$, then the same subset in every corresponding output graph has rank at most $\ell$. Intuitively speaking, such a transduction can only forget information from the input graph. 
To the best of our knowledge, the notion of rank-decreasing transductions is new.  It is a simple notion, and yet it seems to be important. 

We are now ready to state the main result of this paper, Theorem~\ref{thm:main-graphs}:  a transduction is rank-decreasing if and only if it is contained in an \mso transduction. This result is proved under the assumption that the input graphs of the transduction have bounded treewidth, and we give counterexamples which show that the result can fail without this assumption. We conjecture that the result can be extended to graphs of bounded cliquewidth, and structures beyond graphs, such as structures with ternary relations. 

As mentioned before, one of the reasons to characterise \mso transductions in non-logical terms is to understand what features of the logic are essential. One such feature is modulo counting. 
In our main theorem, we use a variant of \mso that allows modulo counting, and we present evidence that this variant is necessary, i.e.~transductions without modulo counting are not sufficient to capture all rank-decreasing transductions. This  situation is broadly reminiscent of the  Seese conjecture~\cite[Sec 4]{seese1991structure}, which in its present state boils down to the importance of counting (the  conjecture has been proved with modulo counting~\cite[Thm 5.6]{courcelle2007vertex}, but remains open without modulo counting).

The paper is organised as follows. Section~\ref{sec:overview} contains an overview, with the minimal definitions  needed to understand the main result of the paper. The overview is phrased in terms of graphs, but later in the paper we develop the framework in terms of general structures (although the main result is proved only for graphs). The main result is proved in  Sections~\ref{sec:measure} and~\ref{sec:rank-decreasinng}. In Section~\ref{sec:measure}, we show the implication 
\begin{align*}
\text{definable in \mso} \quad \Rightarrow \quad \text{rank-decreasing}.
\end{align*}
 We do this by giving several descriptions of rank, which have an increasingly logical character; with the last description the implication $\Rightarrow$ becomes self-evident. In Section~\ref{sec:rank-decreasinng}, we prove the converse implication $\Leftarrow$ (under the assumptions that the input graphs have bounded treewidth; this assumption is not needed for $\Rightarrow$). We first prove 
$\Leftarrow$ for transductions that input trees, and then we extended to inputs of bounded treewidth by using the definable tree decompositions from~\cite[Thm 2.4]{bojanczykDefinabilityEqualsRecognizability2016a}. In Section~\ref{sec:further}, we present some conjectures about how our main theorem could be extended beyond graphs of bounded treewidth, and we discuss some consequences of these conjectures. Finally, in Appendix~\ref{sec:bounded-branching} we make some limited progress on these conjectures. Although Appendix~\ref{sec:bounded-branching} is not the main conceptual contribution of this paper, it is the most technically involved part.

%% file: overview.tex
\section{Overview}
\label{sec:overview}
In this section, we describe the minimal definitions that are needed to understand the main result of the paper. In later sections, these definitions are generalized, e.g.~from graphs to structures that allow relations of arity bigger than two. These  generalizations will be used both to prove the main result and to state conjectures.

\subsection*{Cut-rank} We begin by describing the cut-rank parameter of Oum, which is a measure that associates a number to every subset of vertices in a graph. 
Consider an undirected graph. For a subset $X$ of the vertices, consider the following matrix which has values in $\set{0,1}$. The rows are vertices $x \in X$, the columns are vertices $y  \not \in X$, and the value of the matrix corresponding to a pair $(x \in X, y \not \in X)$ is $0$ or $1$ depending on whether there is an edge connecting the two vertices. The \emph{cut-rank} of the subset $X$ is defined to be the rank of this matrix (i.e.~number of linearly independent rows, or equivalently, columns) when viewed as a matrix over the two-element field. 

\begin{example}\label{ex:rank-and-size}
    The cut-rank of a subset is at most its size, because a matrix with $n$ rows can have rank at most $n$.
\end{example}

\begin{example}
 Consider a graph that is a path together with a connected subset $X$:
 \mypic{16}
The corresponding matrix has at most two rows that contain exactly one $1$, and the remaining rows contain only $0$. Hence, the cut-rank is at most $2$. More generally, the cut-rank of a subset in the path is proportional to the number of edges that are cut. 
\end{example}

\begin{example}[Cut-rank in a grid]\label{ex:cut-rank-in-grid}
 Consider a graph that is a square $n \times n$ grid, as in the following picture: 
 \mypic{14}
 If a subset $X$ of vertices in this grid contains at most half  of all the vertices, then its cut-rank is approximately its size, as expressed in the following formula: 
 \begin{align*}
 \sqrt{|X|} 
 \quad \le \quad 
 \text{cut-rank of $X$} 
 \quad \le \quad 
 |X|.
 \end{align*}
Since the cut-rank of a set is the same as the cut-rank of its complement, it follows that for every subset in the square grid, its cut-rank is approximately the same as minimum of $\{$size of the subset, size of its complement$\}$. 
\end{example}

\subsection*{Transductions} The other main topic of this paper is transductions, which are transformations that input and output graphs, or more general structures.  We begin with a very general notion of transduction.

\begin{definition} 
 A \emph{graph-to-graph transduction} is any binary relation between graphs, such that for every pair in the transduction, both graphs have the same vertices.
\end{definition}

The assumption that the two graphs have the same vertices will allow us to relate vertices in two graphs\footnote{The ability to relate vertices in the input and output graphs (which is implemented here by simply using the same vertices in both graphs) is similar to the origin information in~\cite{bojanczykTransducersOriginInformation2014}.
}. We will be mainly interested in comparing  the cut-rank of the same subset of vertices in both graphs, see Definition~\ref{def:rank-decreasing-graph-to-graph}.
In the paper, we use the following terminology: if a transduction contains a pair of graphs $(G,H)$, then $G$ is called the input graph and $H$ is called the output graph. 
We do not even make any assumptions about the transduction being closed under isomorphism, although all examples considered in this paper will be closed under isomorphism in the following way: if we take a pair of graphs that is related by a transduction, and apply the same bijection to the vertices of both graphs, then the new pair is also in the transduction.

\begin{example}[Examples of transductions]\label{ex:examples-of-transductions}
 Here are some examples of graph-to-graph transductions.
 \begin{enumerate}
 \item \label{it:remove-edges-transduction} the output graph is obtained from the input graph by removing any subset of edges;
 \item \label{it:remove-edges-transduction-btw} as in~\ref{it:remove-edges-transduction}, but also the input graph has treewidth $\le 7$.
 \end{enumerate} 
 For each of the above examples, there is also a converse transduction, where the input and output graphs are swapped. Taking the converse will not preserve some properties of graph-to-graph transductions that will be considered in this paper, such as being definable in \mso or being rank-decreasing.
\end{example}

\subsection*{Rank-decreasing transductions} We now present the main definition of this paper, which is rank-decreasing transductions. The idea is that the cut-rank of a subset in the output graph is bounded by a function of the cut-rank for the   same subset in the input graph.

\begin{definition}
 [Rank-decreasing transduction] \label{def:rank-decreasing-graph-to-graph} A graph-to-graph transduction is called \emph{rank-decreasing} if there is some function $f : \Nat \to \Nat$ such that 
 \begin{align*}
 \text{cut-rank of $X$ in $H$} \le f(\text{cut-rank $X$ in $G$}) 
 \end{align*}
 for every $(G,H)$ in the transduction and subset of vertices $X$.
\end{definition}

In the above definition, we do not need to specify if we refer to vertices of the input or output graph, since both graphs have the same vertices. The rank does not need to strictly decrease when applying the transduction, so it might be more formal to call it rank-non-increasing, but we use the name rank-decreasing for brevity.

\begin{example}\label{ex:grid-counterexample}
 Consider the transduction from item~\ref{it:remove-edges-transduction} in Example~\ref{ex:examples-of-transductions}, which removes edges from a graph. If the input graph is a clique, then every subset has cut-rank $1$. However, after removing some edges, the same subset can have an arbitrarily large cut-rank. Therefore, this transduction is not rank-decreasing. The same kind of  argument applies to the converse transduction, which adds edges. 

 Consider the transduction from item~\ref{it:remove-edges-transduction-btw}, which removes edges, but where the input graphs are restricted to treewidth $\le 7$. As we will explain later in this paper, this transduction is definable\footnote{The transduction is clearly definable in the model which allows quantification over subsets of the edges, but it is also definable with quantification over subsets of the vertices when the inputs have bounded treewidth, which follows from a result called the Sparseness Theorem~\cite[Section 9.4]{courcelleGraphStructureMonadic2012}. } in monadic second-order logic, and therefore it is rank-decreasing. Consider now the converse of this transduction, i.e. the transduction that adds edges, but subject to the assumption that output graph has treewidth $\le 7$. This transduction is not rank-decreasing. Indeed, suppose that the input graph has no edges. Then every subset of vertices has cut-rank $0$. However, if the output graph is a path, which has treewidth 1, then there can be subsets of vertices in the output graph that have arbitrarily large cut rank. 
\end{example}

\begin{example}\label{ex:input-grids-rank-decreasing}
    Consider any transduction where every input graph is a grid. Then, by Examples~\ref{ex:rank-and-size} and~\ref{ex:cut-rank-in-grid}, this transduction is necessarily rank-decreasing.
\end{example}

\subsection*{Definable transductions} In this paper, we will be mainly interested in transductions that can be defined in logic. The logic that we use is \cmso, which is  monadic second-order logic equipped with modulo counting (i.e.~one can check if the size of a set is divisible by two, or by three, etc.). We assume that the reader is familiar with the syntax and semantics of monadic second-order logic, see~\cite[Section 2]{Thomas97}. In this paper, we model a graph as a structure where the universe is the vertices, and the edges are represented by a binary relation (which is symmetric for undirected graphs).

\begin{definition}[Definable graph-to-graph transductions]
 The syntax of a \emph{definable graph-to-graph transduction} is given by:
 \begin{itemize}
 \item A set of \emph{colours}, which is a finite set.
 \item An \emph{edge formula}, which is a formula of \cmso that has two free variables, and uses the vocabulary of graphs (a single binary relation), extended with a family of unary predicates that has one unary predicate per colour. 
 \end{itemize}
 \end{definition}

 The semantics of the transduction is defined as follows. Let $G$ be a graph, which will be treated as an input graph. A \emph{colouring} is a function that maps vertices of the graph to colours from the transduction. Every colouring induces a structure, in which apart from the edge relation there is a unary predicate for each colour. The output graph is then defined over the same vertices as the input graph, with edges defined by the edge formula. (The output graph should be undirected, i.e. the edge formula should define a symmetric relation for every input graph and colouring; otherwise, the transduction is ill-formed.) Note that there is one output graph per colouring, some of which may coincide, and therefore if the input graph has $n$ vertices and there are $k$ colours, then there are at most $k^n$ output graphs. 

In the literature, transduction usually have additional features, namely, copying, removing vertices, and filtering.
Copying and removing vertices are irrelevant here because we will only consider transductions which preserve the set of vertices.
Likewise, filtering can be omitted for our purposes because we will work with sub-definable transductions (see below).
For these reasons, we prefer to work with the slightly simpler notion of transductions presented above.

\subsection*{Rank-decreasing equivalent to sub-definable}
We now present the main result of the paper, which is that for graph-to-graph transductions with inputs of bounded treewidth, being rank-decreasing is the same as being definable. There is a small caveat: since being rank-decreasing is closed under removing (input, output) pairs from a transduction, a similar closure must be applied to definable transductions. 

\begin{definition}[Sub-definable]
 A graph-to-graph transduction is called \emph{sub-definable} if it is contained, as a binary relation on graphs, in some definable transduction.
\end{definition}

For example, since the identity transduction is definable, it follows that every subset of the identity is sub-definable, e.g.~the identity on graphs restricted to graphs with a prime number of vertices.
Here is the main theorem of this paper. 

\begin{theorem}\label{thm:main-graphs} Let $\alpha$ be a graph-to-graph transduction with inputs of bounded treewidth.
 Then $\alpha$ is rank-decreasing if and only if it is sub-definable.
\end{theorem}

The right-to-left implication is true without the assumption on bounded treewidth, see Corollary~\ref{cor:definable-is-rank-decreasing}. The assumption is used for the left-to-right implication.
The proof of the theorem is described in the rest of this paper. 
Our proof does not rely on the outputs being graphs, and it would work for other classes of output structures, such as hypergraphs or matroids. It would  also work for transductions that input more general structures, e.g.~with ternary relations, subject to the bounded treewidth restriction (using the usual generalisation of treewidth to structures, in terms of  the Gaifman graph).
We conjecture that the result can be extended to transductions with inputs of bounded cliquewidth (or, equivalently, bounded rankwidth), and also inputs more general than graphs, see Section~\ref{sec:further} for more information. However, some restrictions on the input are necessary, as explained in the following example.

\begin{example} Recall that in Example~\ref{ex:input-grids-rank-decreasing} we have shown that every transduction that inputs grids is rank-decreasing. We will now show that there is a transduction that inputs grids and which is not sub-definable, and therefore this transduction will be a witness for the failure of the implication 
    \begin{align*}
    \text{rank-decreasing}
    \quad 
    \not \Rightarrow
    \quad 
    \text{sub-definable.}
    \end{align*}

Consider the graph-to-graph transduction that consists of all pairs $(G,H)$ where $G$ is a square grid and  $H$ is a graph on the same set of vertices. If the input grid has $n$ vertices (this $n$ is necessarily a square number), then the number of possible outputs is $2^{n \choose  2}$. In the limit, this number is bigger than $k^n$ for every $k$, which is the maximal number of outputs for a transduction that uses $k$ colours. 
\end{example}

The main theorem is proved in the following two sections.


%% file: ranks.tex
\section{Sub-definable implies rank-decreasing}
\label{sec:measure}
In this section, we prove the easier implication in the theorem, which says that every definable transduction is necessarily rank-decreasing. (This implies that every sub-definable transduction is rank-decreasing since rank-decreasing transductions are downward closed.) This implication does not need the assumption that the inputs have bounded treewidth. 

\input{structures}

\subsection{Cut-rank for general structures}
We begin by generalizing cut-rank from undirected graphs to general structures. This will be useful later on in the paper, where we measure the cut-rank in structures that do not have a clear graph representation.

Consider a structure over some vocabulary. The \emph{quantifier-free type} of a $k$-tuple of elements in the universe is defined to be the set of quantifier-free formulas with $k$ free variables that are true for the tuple. If the vocabulary is fixed and the number of variables $k$ is fixed, then there are finitely many possible quantifier-free types. This is because we use finite vocabularies.

\begin{definition}[Cut-rank and type matrix]\label{def:cut-rank-structures-and-types}
 Let $\structa$ be a structure whose vocabulary has maximal arity $m$. Define the \emph{type matrix} of a subset $X$ of its universe to be the following matrix: 
 \begin{enumerate}
 \item {\bf Rows.} Tuples $(x_1,\ldots,x_m) \in X^m$.
 \item {\bf Columns.} Tuples $(y_1,\ldots,y_m) \in (\structa \setminus X)^m$.
 \item {\bf Values.} Quantifier-free type $\bar x \bar y$.
 \end{enumerate}
 The \emph{cut-rank} of a subset of the universe $X$ is defined to be the number of distinct rows in this matrix.
\end{definition}

In Lemma~\ref{lem:rank-row-column}, we will justify that essentially the same notion of cut-rank will arise if count the number of distinct columns, or if we compute the rank under a suitable representation of quantifier-free types as elements of a finite field. Furthermore, in the case of structures which are graphs, all of these notions are essentially the same as the version of cut-rank from Section~\ref{sec:overview}.  First, however, we explain what we mean by ``essentially the same notion''.

\subsection{Asymptotic equivalence}
We will consider two notions of rank to be equivalent if they are bounded by each other. In this section, we introduce a notation for boundedness that will be used in the paper. 

Consider a number-valued function, i.e.~a function that assigns natural numbers to inputs from some domain $X$. The example that will be relevant here is when the domain $X$ is the set of pairs (structure, subset of its universe) and the function is the cut-rank. We say that this function is bounded on a subset of the domain if its outputs on that subset have a finite upper bound. Consider two number-valued functions $\mu_1,\mu_2$ with the same domain $X$. We say that $\mu_1$ is \emph{bounded} by $\mu_2$ if every subset $X \subseteq Y$ satisfies
\begin{align*}
\text{$\mu_2$ bounded on $Y$} \quad \Rightarrow \quad \text{$\mu_1$ is bounded on $Y$}.
\end{align*}
We say that the functions are \emph{asymptotically equivalent} if they are bounded by each other.


\subsection{Equivalent variants of cut-rank}

Define a \emph{rank function} over a class of structures to be a function that assigns a natural number in $\set{0,1,\ldots}$ to every pair consisting of  a structure from the class and a subset  of its universe. Cut-rank is one example. The following lemma and its proof justify the choice made in Definition~\ref{def:cut-rank-structures-and-types} by showing that other choices are equivalent. 

\begin{lemma}\label{lem:rank-row-column}
 For every class of structures, the following rank functions are asymptotically equivalent:
 \begin{enumerate}
 \item the number of distinct rows in the type matrix;
 \item the number of distinct columns in the type matrix;
 \item the (algebraic) rank of the type matrix, when its values are embedded arbitrarily into a finite field.
 \end{enumerate}
\end{lemma}
\begin{proof}
    Recall that the type matrix takes values in the set of quantifier-free types with $2m$ variables; let $t$ denote the number of such types and $c$ and $r$ respectively denote the number of columns and rows of the matrix.
    Then each row is a sequence of $c$ quantifier-free types with $2m$ variables, therefore there are at most $t^c$ different rows, which is asymptotically equivalent to $c$.
    Likewise, $t \leq r^c$.

    Finally, it is a well known fact that the algebraic rank is comprised between $\Omega(\log r)$ and $r$.
 \end{proof}

 Note that the type matrices of a set and its complement are transpositions of each other; therefore it follows from the lemma that both ranks are asymptotically equivalent.

\subsubsection*{Formulas with quantifiers}
When defining cut-rank, we use the type matrix which stores quantifier-free types. One could also consider a version of the type matrix which stores types for formulas with quantifiers. For $d \in \set{0,1,\ldots}$, define the \emph{\cmso theory of quantifier depth $d$} of a $k$-tuple of elements in a structure to be the set of \cmso formulas that are true for the tuple, and which use at most $d$ nested quantifiers. The \emph{$d$-type matrix} is defined in the same way as the type matrix, except that the values are \cmso theories of quantifier depth $d$. In particular, the 0-type matrix  is the original notion of type matrix.

\begin{theorem}\label{thm:rank-invariant-under-quantifier-rank}
 For every class of structures and $d \in \set{0,1,\ldots}$, cut-rank as defined in Definition~\ref{def:cut-rank-structures-and-types} is asymptotically equivalent to the function that counts the number of distinct rows in the $d$-type matrix. 
 \end{theorem}

The theorem can be derived as an application of Ehrenfeucht-Fraïssé games; this is presented in Appendix~\ref{app:ef-games}.

We have now collected all the ingredients needed to prove the easier  implication of our main theorem, even without any assumption on inputs having bounded treewidth. 
\begin{corollary}\label{cor:definable-is-rank-decreasing}
 Every sub-definable transduction is  rank-decreasing.
\end{corollary}
\begin{proof}
 Clearly, transductions which are rank-decreasing are closed under composition. Every definable transduction is a composition of three definable transductions: (1) add some colouring to the universe; (2) extend the vocabulary with all relations that can be defined using definable formulas with at most $k$ variables and quantifier depth at most $d$; (3) remove some relations from the vocabulary. It is easy to see that step (1) is rank-decreasing. Step (2) is rank-decreasing thanks to Theorem~\ref{thm:rank-invariant-under-quantifier-rank}. Finally, step (3) is rank-decreasing. Let us point out that a strict decrease of the ranks can only occur in step (3).
 \end{proof}

\subsubsection*{Rank functions in selected classes of structures.} So far, we have discussed asymptotic equivalence for rank functions over all structures for a given vocabulary. We finish this section with  some examples that use restricted structures, such as linear orders or equivalence relations.

\begin{example}[Rank in undirected graphs and grids]\label{ex:graph-rank} For undirected graphs, cut-rank as in Definition~\ref{def:cut-rank-structures-and-types} is asymptotically equivalent to the original notion of cut-rank that was  described in Section~\ref{sec:overview}. As discussed in Example~\ref{ex:cut-rank-in-grid}, in the class of grids, cut-rank is equivalent to the minimum of (size of the set, size of its complement).
 \end{example}

\begin{example}[Rank in linear orders]\label{ex:linear-orders-rank}
 Consider the class of linear orders from Example~\ref{ex:linear-orders}. Over this class, cut-rank is asymptotically equivalent to the function that maps a subset $X$ of the universe  to the minimal number $k$ such that $X$ is a union of $k$ intervals. 
\end{example}

\begin{example}[Rank in equivalence relations]\label{ex:equivalence-relations}
 Consider the class of equivalence relations, i.e.~undirected graphs that are disjoint unions of cliques. Over this class, cut-rank is asymptotically equivalent to the measure that maps a set to the number of equivalence classes that are cut by it (i.e.~equivalence classes that contain elements in the set and also outside the set).\end{example}

The following lemma is often convenient.

\begin{lemma}\label{lem:rank_of_union}
The cut-rank of $X \cup Y$ is asymptotically bounded by the maximal cut-rank of $X$ and $Y$.
\end{lemma}

\begin{proof}
Assume that there are $\ell$ different rows in the type matrix of $X \cup Y$, we should prove that there are at least $f(\ell)$ different rows in the type matrix of either $X$ or $Y$, where $f$ goes to infinity.
Without loss of generality, we assume that the different rows in the matrix are labelled with tuples $\bar z_i = \bar x_i \bar y_i$, where $\bar x_i$ is a tuple of elements of $X$ and $\bar y_i$ is a tuple of elements of $Y$.
By definition, for each $i<j$, there is an tuple $\bar c$ of elements outside $X \cup Y$ such that the quantifier-free types of $\bar x_i\bar y_i \bar c$ and $\bar x_j \bar y_j \bar c$ differ.

This implies that the quantifier-free type of $\bar x_i \bar y_j \bar c$ differs with either the one of $\bar x_i \bar y_i \bar c$ or of $\bar x_j \bar y_j \bar c$.
By Ramsey's theorem, we may assume that the same case holds for all $i<j$, say, the first one (they are symmetrical), up to reducing the number of rows to $f(\ell)$.
Then the fact that the quantifier-free types of $\bar x_i \bar y_j \bar c$ and $\bar x_j \bar y_j \bar c$ implies that the type matrix of $X$ has $f(\ell)$ different rows.
\end{proof}

A consequence of the lemma is that the cut-rank of a boolean combination of $n$ subsets of bounded cut-rank is asymptotically bounded by $n$.

%% file: structures.tex
\subsection{Structures and logic}
\label{sec:structures-and-logic}
Although the overview in Section~\ref{sec:overview} discussed undirected graphs, we will consider other structures in this paper, e.g.~directed graphs or structures with a ternary relation. 
We begin by fixing some logical terminology. 
A \emph{vocabulary} is defined to be  a set of relation names, each one with a distinguished \emph{arity} in $\set{1,2,\ldots}$, which says how many arguments the relation takes. A \emph{structure} over such a vocabulary consists of a set, called the \emph{universe of the structure}, together with an interpretation which maps each relation name in the vocabulary to a relation on the universe of corresponding arity. All structures in this paper are finite, which means that the universe is finite.  The vocabularies are also finite. When we talk about a \emph{class of structures}, we assume that all structures in the class have  the same vocabulary, and that the class is closed under isomorphism of structures. Here are some classes of structures that will appear in this paper. 

\begin{example}[Orders]\label{ex:linear-orders} In the class of \emph{linear orders}, the vocabulary has one binary relation $x \le y$, and the class consists of all structures where this binary relation is a linear order, i.e.~antisymmetric, total, reflexive and transitive. 
\end{example}

\begin{example}[Directed graphs]\label{ex:graph-representation}
    The class of \emph{directed graphs} is the class of all structures over a vocabulary with one binary relation, without any restrictions such as symmetry.
\end{example}

To define properties of structures, we use monadic second-order logic with counting, denoted by \cmso.



\subsubsection*{Transductions} In Section~\ref{sec:overview}, we have defined the  special case of graph-to-graph transductions. The extension to arbitrary structures is defined in the obvious way. If $\Cc$ and $\Dd$ are classes of structures, then a $\Cc$-to-$\Dd$ transduction is any set of pairs of structures, where the inputs are from the class $\Cc$ and the outputs are from the class $\Dd$, and furthermore the input structure has the same universe as the output structure.  Recall that a class of structures is required to use a single vocabulary, and therefore there will be two vocabularies involved in a transduction: the input vocabulary and the output vocabulary. 

The notion of definable transductions is lifted to structures in the obvious way. The formulas in the transduction use the input vocabulary, and instead of having a binary edge formula, there is one formula with $n$ free variables for every relation name in the output vocabulary that has $n$ arguments.

%% file: tree-transductions.tex
\section{Rank-decreasing implies sub-definable}
\label{sec:rank-decreasinng}
We now turn to the more difficult implication of Theorem~\ref{thm:main-graphs}. Here,  we use the assumption on inputs having bounded treewidth. We prove this implication in three steps: first for inputs that have no structure except for equality, then for tree inputs, and finally for inputs of bounded treewidth.

\input{independent}
\input{node-coloured}

\subsubsection{Transductions that input trees}
\label{sec:finish-proof-of-from-trees}
Having defined trees and their representation, in  this section, we prove that the implication 
\begin{align*}
\text{rank-decreasing} \quad \Rightarrow \quad \text{sub-definable}
\end{align*}
holds for all transductions that input trees, with no assumptions on the output structures. In the following, a tree-to-$\classc$ transduction is a transduction that inputs trees, with the representation discussed in the previous section, and outputs structures from a class $\classc$. The output class $\classc$ need not be a class of graphs; it can use relations of arity more than two.  

\begin{theorem}\label{thm:from-trees}
 Let $\classc$ be a class of structures. If a tree-to-$\classc$ transduction is rank-decreasing, then it is sub-definable. 
\end{theorem}

This theorem is proved using a combination of the results from Sections~\ref{sec:independent-sets} and~\ref{sec:trees}. The general idea is that a pair $(\structt, \structa)$ in a rank-decreasing transduction can be seen as a tree decomposition of $\structa$, where every subtree of the tree $\structt$ describes a part of the structure $\structa$, and this part is further partitioned by the children of the subtree\footnote{Some parts of our proof are inspired by ideas dating back to~\cite{courcelle1995logical}.}.

We will use a tree automaton to reconstruct the 
structure $\structa$ by a bottom-up pass through this tree. The run of the automaton will be represented using the coloured trees from Section~\ref{sec:trees}, while the local update which uses the states in the child subtrees to compute the state in their parent will be handled using the results from Section~\ref{sec:independent-sets}. 

\subsubsection*{Induced sub-structures.} We begin by explaining how a part of a structure is viewed as a structure. The idea is that we want to retain information about how the sub-structure communicates with its exterior. This is formalized in the following definition.

\begin{definition}
 Let $\structa$ be a structure with maximal arity $m$ and let $X$ be a subset of its universe. The structure $\structa | X$ is defined as follows. The universe is $X$. For every $\ell \leq m$, every $c_1,\ldots,c_\ell \in \structa \setminus X$ and every quantifier-free formula 
 \begin{align*}
 \varphi(x_1,\ldots,x_k, y_1,\ldots,y_{\ell} )
 \end{align*}
 with up to $m$ variables over the vocabulary of $\structa$, the structure $\structa | X$ has a relation of arity $k$ which is interpreted as 
 \begin{align*}
 \setbuild{(a_1,\ldots,a_k) \in X^k} { $\structa \models \varphi(a_1,\ldots,a_k, c_1,\ldots,c_\ell)$}.
 \end{align*}
\end{definition}

The vocabulary of the structure in the above definition is unbounded and depends on the elements of the structure that are not in $X$. This will not be an issue, since we will be interested in induced sub-structures of bounded rank, and for such sub-structures the vocabulary will be essentially bounded, as explained in the following lemma. 

\begin{lemma}
    Let $\classc$ be a class of structures with maximal arity $m$. For $\structa \in \classc$ and $X$ is a subset of its universe, the number of quantifier-free types with $m$ variables in $\structa | X$ is asymptotically equivalent to the rank of $X$ in $\structa$. 
\end{lemma}

Induced sub-structures are well behaved with respect to cut-rank as witnessed by the following lemma.

\begin{lemma}\label{lem:rank_in_substructure}
    For any structure $\structa$ and $X$ and $Y$ two subsets of its universe such that $Y \subseteq X$, the cut-ranks of $Y$ in $\structa$ or in $\structa | X$ are the same.
\end{lemma}


Consider a structure $\structa$, and a partition of its universe into subsets $X_1,\ldots,X_n$. 
In the proof below, we will be interested in how the type of a tuple $\bar a$ in  $\structa$ can be deduced from the type of this tuple in the induced sub-structures $\structa|X_1,\ldots,\structa|X_n$. Since the induced sub-structures might not contain some elements of the tuple, we will want to talk about partial tuples, defined as follows. 
A \emph{partial $n$-tuple} in a structure is a tuple of length $n$, where every component is either an element of the structure, or an ``undefined'' value. Such partial tuples arise when restricting a structure to a subset: if $\bar a$ is a (complete) $n$-tuple in $\structa$, and $X$ is a subset of the universe in $\structa$, then $\bar a$ can be viewed as a partial $n$-tuple in $\structa | X$, which is undefined on those coordinates that use elements not in $X$. The type of a partial $n$-tuple is defined in the same way as the type of a complete $n$-tuple, except that it stores information about which coordinates are defined/undefined.  The following lemma is an immediate corollary of the definitions. 

\begin{lemma}
    [Compositionality] Let $\structa$ be a structure, and let $\Xx$ be a family of subsets that partition its universe. The type of a tuple $\bar a$ in $\structa$ depends only on the types of this tuple projected into the structures $\structa|X$ for $X \in \Xx$.
\end{lemma}

We now rephrase the Compositionality Lemma in more precise terms, preparing the ground for the proof below.

\begin{lemma}\label{lem:transition-function-on-types}
 For every fixed vocabulary and $k,\ell, m \in \set{1,2,\ldots}$ 
 one can find a finite set of colors $Q$ with the following property. Let $\structa$ be a structure over the fixed vocabulary together with a family of subsets $\Xx$ that partitions its universe, such that every subfamily $\Yy \subseteq \Xx$ has union of  rank  at most $k$. Then there is a family of functions 
 \begin{align*}
\lambda_X : \text{types of partial $m$-tuples in $\structa|X$} \to Q, 
 \end{align*}
 with one function for every part $X \in \Xx$, and a function
 \begin{align*}
 \gamma :\ & Q^\ell \to \text{types of partial $m$-tuples in $\structa$}
 \end{align*}
 such that for every distinct parts $X_1,\ldots,X_\ell \in \Xx$ and every partial $m$-tuple $\bar a$ in $\structa$ that is contained in $X_1 \cup \cdots \cup X_\ell$, the type of $\bar a$ in $\structa$ is obtained by applying $\gamma$ to 
 \begin{align*}
 (\lambda_{X_1}(\text{type of $\bar a$ in $\structa | X_1$}), \ldots, 
 \lambda_{X_\ell}(\text{type of $\bar a$ in $\structa | X_\ell$})).
 \end{align*}
\end{lemma}

We now have all the ingredients that are needed to prove Theorem~\ref{thm:from-trees}. 
Consider a pair $(\structt, \structa)$ that belongs to a rank-decreasing tree-to-$\classc$ transduction. For every node $X$ of the tree $\structt$, we will be interested in the structure $\structa | X$ along with the partition $\Xx_X$ that corresponds to the partition of $X$ into its children subtrees. The following lemma shows that this equivalence relation satisfies the assumptions of Lemma~\ref{lem:transition-function-on-types}.

\begin{lemma}
 There is some $k$ such that for every $(\structt, \structa)$ in the transduction, every node $X$ of $\structt$, and every subfamily $\Yy \subseteq \Xx_X$, the rank of the union of $\Yy$  in $\structa|X$ is at most $k$. 
\end{lemma}

\begin{proof}
 By the assumption that the transduction is rank-decreasing. (to-do: expand this justification) 
\end{proof}

Apply Lemma~\ref{lem:transition-function-on-types} to the $k$ from the above lemma, and $m$ being the maximal arity of relations in the vocabulary of the class $\classc$. For every node $X$ in the tree, apply Lemma~\ref{lem:transition-function-on-types} to the structure $\structa | X$ and the partition of $X$ into its children, yielding some functions 
\begin{align}
\label{eq:transition-functions}
 \set{\lambda_Y}_{\text{$Y$ is a child of $X$}} 
 \quad \text{and} \quad 
 \gamma_X.
\end{align}
The family of functions described above can be viewed as a node-coloured tree $\structs$, where the colour of each node $X$ indicates the functions $\lambda_X$ and $\gamma_X$. We say that $\structs$ follows from applying Lemma~\ref{lem:transition-function-on-types} to $(\structt, \structa)$ if it arises as a result of the above construction. The following lemma shows that the structure $\structa$ can be recovered from the node-labelled tree $\structs$ using a sub-definable transduction.

\begin{lemma}
 The following transduction is sub-definable 
 \begin{align*}
 \setbuild{(\structs,\structa)}{$\structs$ follows from applying Lemma~\ref{lem:transition-function-on-types} to $(\structt,\structa)$ for \\ 
 some $\structt$ such that $(\structt, \structa)$ is in the transduction }.
 \end{align*}
\end{lemma}
\begin{proof}
 The functions $\lambda_X$ and $\gamma_X$ that are stored in the tree $\structs$ give an inductive procedure to compute the type of an $n$-tuple of elements in the structure $\structa$, which only uses the nodes and their colors in $\structs$. 
 For an $m$-tuple $\bar a$ of elements in $\structa$ (or equivalently, in $\structs$, because these have the same universe), define its \emph{run} to be the labelling 
 \begin{align*}
 \text{nodes of $\structs$} & \to Q\\
 X & \mapsto \lambda_X(\text{type of $\bar a$ in $\structa|X$}).
 \end{align*}
 By Lemma~\ref{lem:transition-function-on-types}, this run has an inductive structure: for every node $X$, its label in the run depends only on the labels of its children, more specifically the labels of the constant number of children that contain elements of the tuple $\bar a$. This inductive procedure can be implemented in \mso, even without modulo counting, by guessing the run in the usual way (guess the run, and then check that it is correctly computed in every node), since \mso is allowed to quantify over colourings of the nodes in a tree. 
\end{proof}

Theorem~\ref{thm:from-trees} follows by pre-composing the transduction from the above lemma with the transduction that maps a tree $\structt$ to all possible colourings $\structs$ that use colors $Q$. The latter transduction is sub-definable thanks to Lemma~\ref{lem:node-labelled-trees}. 

\subsubsection{A criterion for being rank-decreasing.} We finish this section with a criterion for checking if a transduction that inputs trees is rank-decreasing. In principle, one should check if for every $k$ there is some $\ell$ such that sets with rank $k$ in the input have rank $\le \ell$ in the output. It turns out that if the transduction inputs trees, then it is enough to check this for some single $k$. Our criterion says that a transduction that inputs trees is rank-decreasing if and only if the ranks are bounded for images of sub-forests.
\begin{theorem}\label{thm:rank-decreasing-criterion} 
 Let $\classc$ be a class of structures. A tree-to-$\classc$ transduction $\alpha$ is rank-decreasing if and only if the rank is bounded for 
 \begin{align*}
 \setbuild{(\structa,X)}{$X$ is a sub-forest of some $\structt$ with $(\structt,\structa) \in \alpha$ }
 \end{align*}
\end{theorem}
\begin{proof}
    Recall from Lemma~\ref{lem:ranks_in_trees} that the rank in trees of subsets of trees is asymptotically equivalent to the least $m$ such that the subset can be written by a boolean combination of $m$ subforests. Furthermore, a Boolean combination of $n$ sets that have bounded rank will have rank bounded by a function of $n$, thanks to Lemmas~\ref{lem:rank_of_union} and~\ref{lem:rank-row-column}.
\end{proof}

\subsection{Transductions that input graphs of bounded treewidth}
\label{sec:endgame-bounded-treewidth}
In this section, we complete the proof of Theorem~\ref{thm:main-graphs}, by generalizing the results from the previous sections  to transductions that input graphs of bounded treewidth. 

For completeness, we recall the definition of treewidth. This is not really necessary for the present paper, since we use treewidth as a black box, by citing a result about definable tree decompositions from~\cite{bojanczykDefinabilityEqualsRecognizability2016a}. Let $G$ be a graph. A \emph{tree decomposition} of $G$ is a tree, where the leaves are vertices of $G$. The \emph{adhesion} of a subtree $X$ in the tree decomposition is defined to be the vertices that are not $X$, but share an edge with some vertex in $X$. The width of a tree decomposition is the maximal size of an adhesion.  The \emph{treewidth} of a graph is the minimal width of tree decompositions.

    The key technical tool here, apart from the results developed so far, is the existence of definable tree decompositions for graphs of bounded treewidth~\cite{bojanczykDefinabilityEqualsRecognizability2016a}. We state this result in a form that is more appropriate for this paper. In the following theorem, we say that a transduction is \emph{sub-definable in both directions} if both the transduction and its converse are sub-definable. 

\begin{theorem}\cite{bojanczykDefinabilityEqualsRecognizability2016a} \label{thm:encode-decode-treewidth} Let $\classc$ be a class of graphs that has bounded treewidth. Then there is a $\classc$-to-tree transduction that is sub-definable in both directions. 
\end{theorem}
\begin{proof}
 We do not reprove this theorem, but only explain how the original result can be rephrased in the way stated here. 
 The main technical result of~\cite{bojanczykDefinabilityEqualsRecognizability2016a} can be stated as follows: 
 \begin{itemize}
 \item[(*)] for every $k$ there is some $\ell$ and a transduction which inputs a graph of treewidth at most $k$, and outputs the same graph together a tree decomposition of width at most $\ell$. 
 \end{itemize}
 The statement (*) is implicit in~\cite[Theorem 2.4]{bojanczykDefinabilityEqualsRecognizability2016a}, and it is explicit in the survey~\cite[Theorem 6.22]{bojanczyk_recobook}. The transduction (*) is  the transduction from the statement of the theorem, in the $\classc$-to-tree direction.
 
 Let us now argue that the converse of the transduction described in (*) is sub-definable. This is a transduction that inputs trees, and therefore it is enough by Theorem~\ref{thm:from-trees} to show that it is rank-decreasing. Here, we can apply the criterion from Theorem~\ref{thm:rank-decreasing-criterion}: for every sub-forest in a tree decomposition, the corresponding subset in the graph has bounded rank. This is easily seen to be true for tree decompositions of bounded width.
\end{proof}

Using the above theorem, we complete the proof of Theorem~\ref{thm:main-graphs}. Consider some rank-decreasing graph-to-graph transduction $\alpha$ that inputs graphs of treewidth at most $k$. Apply Theorem~\ref{thm:encode-decode-treewidth}, yielding a graph-to-tree transduction $\beta$ whose inputs contain all inputs of $\alpha$, and such that both $\beta$ and its converse $\bar \beta$ are sub-definable. Consider the following diagram:
\[
 \begin{tikzcd}
 \text{graphs}
 \ar[r,"\beta"]
 \ar[d,"\alpha"]
 & 
 \text{trees}
 \ar[dl,"\bar \beta; \alpha"]\\
 \text{graphs}
 \end{tikzcd} 
 \]
Since $\beta;\bar \beta$ contains the identity relation on inputs of $\alpha$, it follows that $\alpha$ is contained in $\beta; \bar \beta; \alpha$. Both $\beta$ and $\bar \beta$ are sub-definable. Therefore, $\bar \beta$ is rank-decreasing by Corollary~\ref{cor:definable-is-rank-decreasing}, and thus the composition $\bar \beta; \alpha$ is rank-decreasing. Since $\bar \beta; \alpha$ inputs trees, it is sub-definable by Theorem~\ref{thm:from-trees}. Thus, $\beta; \bar \beta; \alpha$ is sub-definable, and so is $\alpha$ as its subset. This completes the proof of Theorem~\ref{thm:main-graphs}.

Observe that our proof did not need to assume that the transduction outputs graphs, and it would also work for transductions that output structures from any class.

%% file: independent.tex
\subsection{Inputs with equality only}
\label{sec:independent-sets}
Define a \emph{structure with equality only} to be a structure over the empty vocabulary. In this structure, logic can only use equality on the elements, since we assume that equality is built-in to the logic. We use the name set-to-$\classc$ transduction for any transduction where the input vocabulary is empty, and the output class is contained in $\classc$. The following lemma shows that for such transductions, rank-decreasing is the same as sub-definable. 

\begin{lemma}\label{lem:from-independent-sets} If a set-to-$\classc$ transduction is rank-decreasing, then it is sub-definable.  
\end{lemma}

We actually prove and use a slightly stronger result, which describes the structures that can arise from rank-decreasing transductions that input structures with equality only.  Define an \emph{informative colouring} of a structure $\structa$ to be a function that assigns a colour to each element of the universe, such that for every tuple of distinct elements in the universe, their quantifier-free type depends only on the corresponding tuple of colours.  In the following lemma, a set-to-$\classc$ transduction is called surjective if every structure from $\classc$ arises as an output for at least one input.

\begin{lemma}\label{lem:informative-coloring}
    The following conditions are equivalent for every class of structures $\classc$: 
    \begin{enumerate}
        \item \label{it:rank-decreasing-from-set} there is a surjective  rank-decreasing  sets-to-$\classc$ transduction;
    \item \label{it:universally-bounded-rank} there is some $k$ such that for every structure in $\classc$, all subsets of its universe have  rank at most $k$; 
        \item \label{it:bounded-informative-colouring} there is some $k$ such that every structure in $\classc$ has an informative colouring with at most $k$ colours. 
    \end{enumerate}
\end{lemma}

Lemma~\ref{lem:informative-coloring} immediately implies Lemma~\ref{lem:from-independent-sets}. This is because for every $k$ there is a definable transduction which inputs a set with equality only, and outputs all possible structures that have this set as the universe, and admit an informative colouring with at most $k$ colors.

The rest of Section~\ref{sec:independent-sets} is devoted to proving Lemma~\ref{lem:informative-coloring}. The equivalence \ref{it:rank-decreasing-from-set} $\Leftrightarrow$ \ref{it:universally-bounded-rank} and the implication \ref{it:universally-bounded-rank} $\Leftarrow$ \ref{it:bounded-informative-colouring}  are straightforward exercises left to the reader. We concentrate on the implication  \ref{it:universally-bounded-rank} $\Rightarrow$ \ref{it:bounded-informative-colouring}. 
 We prove this implication   by induction on the maximal arity of relations in the vocabulary of $\classc$. The induction base is when the maximal arity is one, i.e.~all relations have one argument. In this case, there is a straightforward informative colouring: map each element to the unary predicates that it satisfies.  Consider now the induction step: suppose that have proved the lemma for vocabularies with maximal arity $m-1$, and we want to prove it for a vocabulary of maximal arity $m$. 

In the proof, we use the following notion of \emph{twin}: 
we say that two elements $a,a' $ in a structure $\structa$ are \emph{twins}  if
\begin{align*}
\structa \models \varphi(a,b_1,\ldots,b_n) 
\quad \iff \quad 
\structa \models \varphi(a',b_1,\ldots,b_n) 
\end{align*}
for every quantifier-free formula $\varphi$ and elements
\begin{align*}
b_1,\ldots,b_n \in \structa \setminus \set{a,a'}.
\end{align*}
Twins will be useful to get an informative colouring, as explained in the following claim. 
\begin{claim}\label{claim:twin-informative}
    Consider a colouring of a structure such that every two distinct elements of the same colour are twins. Then by, possibly multiplying the number of colours by at most $m$, we can get an informative colouring.
\end{claim}
\begin{proof}
    We increase the number of colours to guarantee the following property (*): there is no colour class whose size is in $\set{2,\ldots,m}$. This can be done by splitting every colour class of bad size into singleton colour classes; the number of colours is at most multiplied by $m$. 
    Consider two non-repeating tuples of elements  $\bar a, \bar b \in \structa^m$ that induce the same  colours coordinate-wise. Thanks to the assumption (*),  one can gradually transform the tuple $\bar a$ into the tuple $\bar b$, in each step replacing one coordinate by another element of the same colour, and getting a non-repeating tuple.  Such a step does not affect the quantifier-free type. Therefore, the quantifier-free types of $\bar a$ and $\bar b$ are the same. 
\end{proof}

Define the \emph{twin graph} of a structure to be the undirected graph where the set of vertices is the universe, and the edge relation is the twin relation. A \emph{twin independent set} is an independent set in this graph, i.e.~a set of elements in the universe that are pairwise unrelated by the twin relation. The following claim uses the assumption from item~\ref{it:universally-bounded-rank} from Lemma~\ref{lem:informative-coloring}, i.e.~that the rank is universally bounded, i.e.~there is some fixed upper bound on the rank of all subsets in all structures from the class. 

\begin{claim}\label{claim:twin-antichains} If a class of structures has universally bounded rank, then   twin independent sets  from $\classc$ have bounded size.
\end{claim}
\begin{proof}
    Consider two elements $a,a' \in \structa$ that are not twins. This means that there is some tuple of elements $\bar b$, in which the elements $a$ and $a'$ do not appear, such that $a \bar b$ and $a'\bar  b$ have different quantifier-free types. Such a tuple is called a \emph{witness} for $a$ and $a'$.
    Observe that the size of a witness can be bounded by the maximal arity $m$ of relations  in the vocabulary. This is because two different quantifier-free types must differ on some relation from the vocabulary. Therefore, if two elements have a witness, then they have a witness that uses $m-1$ elements.

    We prove the claim by contradiction -- if there are arbitrarily large twin independent sets, then the rank is not universally bounded. We will show that if there are arbitrarily large twin independent sets, then there are arbitrarily large twin independent sets with the following \emph{external witness property}:  for every pair of elements in the   set, the corresponding witness is disjoint with the set. A twin independent set  of size $n$ with the external witness property gives a subset of the universe that has  rank $n$, because in the corresponding type matrix all rows are different, with the columns being the elements that appear in the witness. Therefore, by the assumption that all subsets in the class $\classc$ have bounded rank, it will follow that there cannot be arbitrarily large twin independent sets with the external witness property. 
    
    To ensure the external witness property, we will use the Ramsey Theorem.
    Suppose that there are twin independent sets 
\begin{align*}
        a_1,\ldots,a_n \in \structa \in \classc
\end{align*}
of arbitrarily large size $n$.
Formally, for all indices $i<j$, there exists a witness $\bar c_{i,j}$ such that the quantifier-free types of $x_i \bar c_{i,j}$ and $x_j \bar c_{i,j}$ differ.
Let $A_0$ denote the set of elements $a_i$ with $i$ even.
By the Ramsey Theorem, we may extract an arbitrarily large twin independent set which has the following homogeneity property: every subset of size $2m$ has the same type of quantifier-depth $m$, if the elements in the subset are listed according to the order inherited from $a_1,\ldots,a_n$, in the structure augmented by a unary predicate colouring $A_0$.
Thus, we now assuming without loss of generality that the homogeneity property is granted.

We claim that the set $A_0$ has the external witness property.
Let $i<j$ be two even indices.
Without loss of generality, we may write $\bar c_{i,j}=\bar d \bar e$ where $\bar d$ contains only elements of the form $a_i$ with $i$ even, and $\bar e$ contains elements that are not of this form.
By definition of $\bar c_{i,j}$, the quantifier-free types of $x_i \bar d \bar e$ and of $x_j \bar d \bar e$ differ.
Therefore there exists a quantifier free formula $\phi$ such that the formula $\psi = \exists \bar y \phi(\bar z \bar y) \text{ and coordinates of $\bar y$ do not belong to } A_0$ holds for $\bar z = x_i \bar d$ but not for $\bar z = x_j \bar d$.

Now let $\bar d'$ be obtained from $\bar d$ by shifting all indices of coordinates $x_k$ to $x_{k+1}$; note that elements of $\bar d'$ do not belong to $A_0$.
By the homogeneity property, the quantifier-depth $m$ types of $x_i x_j \bar d$ and of $x_i x_j \bar d'$ are the same, therefore $\psi$ holds for $x_i \bar d'$ and not for $x_j \bar d'$.
We conclude that $A_0$ has the external witness property.
\end{proof}

Let $d$ be the maximal size of twin independent sets. If we take a twin independent set in the structure that has maximal size, then it will also be a \emph{twin dominating set}, which is defined to be a set with the following property: every element of the universe is either in the set, or a twin of some element in the set. Therefore, every structure in $\classc$ has a twin dominating set of size at most $d$.

Consider a structure  $\structa \in \classc$ and a twin dominating set of size at most $d$ that is listed  in some order  $\bar c = c_1,\ldots,c_d$. We can use repetitions in the sequence if there are fewer than $d$ elements  in the structure.  Define a new structure, call it $\structa_{\bar c}$ as follows.   For every pair $(i,\varphi)$ where  $i \in \set{1,\ldots,d}$ and $\varphi$ is a  quantifier-free formula  with $m$ variables  over the vocabulary of $\classc$, the structure has a relation of arity $m-1$ which is   interpreted as 
\begin{align*}
\setbuild{ (a_2,\ldots,a_m) \in \structa^{m-1}}{$\structa \models \varphi(c_i,a_2,\ldots,a_m)$}.
\end{align*}
Let $\classd$ be the class of all structures $\structa_{\bar c}$ that arise this way. This class can be obtained from the original class $\classc$ by applying a definable transduction, and therefore it must also have universally bounded rank.  Since the vocabulary of $\classd$ has smaller maximal arity than $\classc$, we can use the induction assumption to conclude that $\classd$ has informative colourings that use a bounded number of colours, say at most $\ell$ colours. 

Take some structure $\structa \in \classc$ and choose a sequence $\bar c = c_1,\ldots,c_d$ that describes  some twin dominating set. Consider a colouring that assigns the following information to an element $a \in \structa$: 
\begin{enumerate}
    \item which of the elements in $\bar c$ are twins of $a$;
    \item the colour in the informative colouring for $\structa_{\bar c} \in \classd$.
\end{enumerate}
The number of colours in this colouring is bounded, since it depends on the constant $d$ from Claim~\ref{claim:twin-antichains} and the bound from the induction assumption on $\classd$. 
To complete the proof of Lemma~\ref{lem:informative-coloring}, we will show that this colouring is informative. By Claim~\ref{claim:twin-informative}, it is enough to show that elements $a$ and $a'$ with the same colour are twins.  We need to show that for every tuple $\bar b$ of length $m-1$, the tuples $a \bar b$ and $a' \bar b$ have the same quantifier-free type. Let $c$ be some element in the twin dominating set $\bar c$ that is a twin of both $a$ and $a'$; such an element must exist since $a$ and $a'$ have the same colour. If $c$ does not appear in the tuple $\bar b$, then we can go from $a \bar b$ to $a' \bar b$ in two steps: first replace $a$ with $c$ and then replace $c$ with $a'$. Both steps keep the quantifier-free type unchanged. If $c$ does appear in the tuple $\bar b$, then use the colouring from the induction assumption on $\classd$ to show that $a \bar b$ and $a' \bar b$ have the same quantifier-free type. This completes the proof of Lemma~\ref{lem:informative-coloring}.

%% file: node-coloured.tex
\subsection{Inputs are trees}
In this section, we prove a special case of the main theorem, in which the inputs are trees. We begin by 
\subsubsection{Trees and their representation}
\label{sec:trees}
One way to think about trees is that these are undirected graphs that are acyclic and connected. This will \emph{not} be the kind of trees that we use in this paper. In the kind that  we do  use, the universe of the structure is only the  leaves of the tree, and the remaining nodes of the tree are represented as subsets of the leaves, as in the following picture: 
\mypic{15}

\begin{definition}[Tree]
    A \emph{tree} over a set   is a family of subsets, called \emph{subtrees}, subject to the following axioms:  (a) every singleton is a subtree; (b) the entire set is a subtree; and (c) if two subtrees are not disjoint, then  one is contained in the other.     
\end{definition}

In the picture before the definition, we did not draw the circles corresponding to singleton subtrees, we follow this convention when drawing trees in the paper to avoid clutter.

    We use the usual tree terminology, such as leaf, parent, child, root, descendant and ancestor. 
    The trees as defined above are unordered, which means that there is no left-to-right ordering on the children of a given node. Also, there are no unary nodes, i.e.~every subtree is either a leaf, or it has at least two children. 

    We will be using transductions that input or output trees, so we need to say how trees are represented using structures. Arguably the most natural way would be to give a monadic  second-order relation $\text{subtree}(X)$. Although such predicates would not break the theory developed here, we choose to stick with first-order relations, and thus we will represent the subtrees using a ternary relation on leaves. The idea is that a pair of leaves $x$ and $y$ specifies a subtree, namely the least subtree that contains both leaves. Since every subtree arises this way, we can use the pairs to represent subtrees, as described in the following definition.  
    
    \begin{definition}[Ternary representation of a tree] For a tree, its \emph{ternary representation} is the structure in which the universe is the leaves and there is a ternary relation defined by 
        \begin{center}
            $z$ is contained in the least subtree containing $x$ and $y$.
        \end{center}
    \end{definition}

    When using transductions that input or output trees, we use the above representation.
    Define a \emph{sub-forest} in a tree to be any set that is obtained by taking some union of sibling subtrees. Sub-forests are easily seen to have bounded cut-rank; in fact, they characterise the cut-rank in trees.
    
    \begin{lemma}\label{lem:ranks_in_trees}
    In a tree $\structt$, the cut-rank of a subset $X$ is asymptotically equivalent to the minimal number of sub-forests of which $X$ is a boolean combination.
    \end{lemma}

    \begin{proof}
        It follows from Lemma~\ref{lem:rank_of_union} together with the observation (Lemma~\ref{lem:rank-row-column}) that cut-ranks are (asymptotically) preserved under complements, that the cut-rank of $X$ is asymptotically bounded by the minimal number $m$ of subforest of which $X$ is a boolean combination.
    
        We now focus on the converse: we should show that if $m$ is large, then the rank of $X$ is large.
        Say that a node $Y$ in $\structt$ is dull if it has a child $Y'$ such that $X \cap (Y \setminus Y')$ is either $Y \setminus Y'$ or the empty set.
        Otherwise, say that $Y$ is interesting.
        
        Let $\ell$ denote the maximal size of a chain of interesting nodes, and $d$ denote the maximal number of interesting siblings.
        It is easy to see that if both $\ell$ and $d$ are small, then $m$ is small; therefore either $\ell$ or $d$ is large.
        
        If $d$ is large then consider a node $Y$ with $d$ interesting children $Y_1,\dots,Y_d$, and for each $i \in \{1,\dots,d\}$, pick an element $x_i \in X \cap Y_i$ and $z_i \in Y_i \setminus X_i$.
        Then for each $i<j$, $z_i$ belongs to the least subtree containing $x_i$ and $z_j$ (which is $Y$), but does not belong to the least subtree containing $x_j$ and $z_j$ (which is $Y_j$), and therefore the quantifier-free types of $x_i z_i z_j$ and $x_j z_i z_j$ differ.
        We conclude that the $x_i$'s define different row in the type matrix, and therefore the cut-rank of $X$ is $\geq d$.
        
        If $\ell$ is large then consider a chain of interesting nodes $Y_1 \subseteq \dots \subseteq Y_m$.
        Note that for each $i\in \{2,\dots,m\}$, since $Y_i$ is interesting and $Y_{i-1}$ is contained in a subtree of $Y_i$, the intersection of $X$ with $Y_i \setminus Y_{i-1}$ is neither full nor empty.
        Pick $x_i \in (Y_i \setminus Y_{i-1}) \cap X$ and $z_i \in (Y_i \setminus Y_{i-1}) \setminus X$.
        As previously, it is easy to see that for $i<j$ the quantifier-free types of $x_i z_i z_j$ and $x_j z_i z_j$ differ, which concludes.        
    \end{proof}

    \paragraph*{Coloured trees.} We will also be using trees where the subtrees are coloured by some fixed set of colours.  One way to represent this colouring is to have, for every colour $c$, a second-order relation $c(X)$ which selects subtrees of colour $c$. If we want to use the ternary representation, then  colouring can be represented by giving for every colour $c$ a binary predicate which selects the pairs of nodes $(x,y)$ such that the least subtree containing $x$ and $y$ has colour $c$. In some constructions of this paper, we will want a transduction or formula to be able to guess a colouring for a tree. It is not immediately clear how this can be done, since the subtrees are not part of the universe of the structure under our chosen representation, and therefore one cannot simply use subsets of the universe to represent the colouring. However, this will be possible, as stated in the following lemma. 

\begin{lemma}\label{lem:node-labelled-trees}
    For every finite set of colours $C$, the following  transduction is definable. 
    \begin{align*}
    \setbuild{
        (\structt, \structs) 
    }
    {$\structt$ is a tree, and $\structs$ is a  $C$-coloured tree \\ that is obtained by some colouring of $\structt$}
    \end{align*}
\end{lemma}

\begin{proof} The proof has two steps.  In the first step, we produce certain coloured trees, which we call orientations. These are used in the second step to produce arbitrary colourings. 

\paragraph*{Step 1. Orientations.}
    Define an \emph{orientation} of a tree $\structt$ to be any colouring $\structs$ of this tree with colours, called ``left'', ``right'' and ``anonymous'',  such that every non-leaf subtree has exactly one child with label ``left'', and exactly one child with label ``right''. We first show that  a definable transduction can map an input tree to at least  one orientation. 

    Consider a colouring of the leaves in a tree with elements from the cyclic group $\Int_3$\footnote{\label{footnote:order-two} To implement the transduction, we will need to use counting modulo three. The same proof works for all orders at least three, and a slightly more elaborate proof works for order two. Therefore, counting modulo any number $k > 1$ would be enough to make the proof work.}. Extend this colouring to subtrees of  $\structt$, by defining the colour of a subtree to be the sum (modulo three) of the colours of the leaves in that subtree. We say that a subtree has a \emph{unique sum} if none of its siblings have the same sum. We will show:
    \begin{itemize}
        \item[(*)] For every tree $\structt$ and  $c \in \Int_3$, one can colour the leaves of $\structt$  with colours from $\Int_3$ so that: (1) the sum of the root subtree is $c$; and (2) every non-leaf has at least two children with unique sum.
    \end{itemize}
    From (*) we easily get an orientation, by defining  the ``left'' and ``right'' children to be the first two children with unique sum, ordered according to their sums. To prove (*), we use a simple induction on the number of leaves. For the induction base, when $\structt$ has one leaf, there is nothing to do. For the induction step, choose two children of the root subtree. To these two children, apply the induction assumption with colours $c$ and $-c$, and for the remaining children of the root apply the induction assumption with colour 0. 

    \paragraph*{Step 2. Arbitrary colourings.} We now use the orientation produced in the previous step  to define an arbitrary colouring. Consider some  orientation $\structs$ of a tree $\structt$. For a subtree $X$, define its \emph{chosen leaf} to be the leaf that is obtained using the following procedure which uses the colours from the orientation: start in the ``left'' child of $X$, and then continue taking ``right'' children until a leaf is reached\footnote{This idea is taken from~\cite[Proposition 3.1]{Thomas97}.}. The \emph{chosen leaf} function is an injective map from (non-leaf) subtrees to leaves, which can be defined using \mso. Therefore, a colouring of the subtrees can be represented using a colouring of the leaves, which can be implemented by a definable transduction since the leaves are the universe. 
\end{proof}

As mentioned in Footnote~\ref{footnote:order-two}, the construction in the above lemma uses modulo counting (although the modulus can be any number $m > 1$).  In fact, we conjecture  modulo counting is necessary. 
\begin{conjecture} The transduction from Lemma~\ref{lem:node-labelled-trees}  cannot be defined without modulo counting, even when restricted to binary trees, i.e.~trees where every subtree has zero or two children.
\end{conjecture}
Since the transduction in the conjecture is easily seen to be  rank-decreasing (for example, because it is sub-definable with modulo counting), the conjecture would imply that modulo counting is necessary to define all rank-decreasing transductions, even for inputs that are binary trees. The conjecture could be seen as evidence in favour of viewing modulo counting as a necessary feature of monadic second-order logic. (A mentioned in the introduction, a related phenomenon appears in the Seese conjecture.)

%% file: extensions.tex
\section{Some conjectures}
\label{sec:further}
In this section, we state some conjectures about how our main theorem could be extended beyond graphs of bounded treewidth. The conjectures in this section are closely related to conjectures that were posed in~\cite{DBLP:journals/corr/abs-2305-18039}.

In the proof of our main theorem, the main external ingredient was Theorem~\ref{thm:encode-decode-treewidth}, about the existence of a graph-to-tree transduction which is sub-definable in both directions, and whose inputs contain all graphs with a given treewidth bound. We conjecture that a similar result holds not just for treewidth, but also for rankwidth (and also for structures beyond graphs, i.e.~for classes of structures that can use relations of arity more than two). We begin by recalling the notion of rankwidth. This is defined in the same way as treewidth, except that we use the rank function to measure the complexity of subtrees in a tree decomposition.

\begin{definition}[Rankwidth] A \emph{tree decomposition} of a structure $\structa$ is a tree where the leaves are the elements of the structure. The \emph{width} of this tree decomposition is the maximal rank of its subtrees. The \emph{rankwidth} of a structure is the minimal width among its tree decompositions.
\end{definition}

The following conjecture is  essentially the same as Conjecture 5.4 in~\cite{DBLP:journals/corr/abs-2305-18039}.

\begin{conjecture}[Main conjecture]\label{conj:main}
 Let $\classc$ be a class of structures of bounded rankwidth. Then there is a surjective tree-to-$\classc$ transduction that is sub-definable in both directions.
\end{conjecture}

As discussed in Section~\ref{sec:endgame-bounded-treewidth}, this conjecture is true for classes of graphs of bounded treewidth. 
Another case where the conjecture is true is graphs of linear rankwidth. This is the special case of rankwidth, where we require the trees to be linear orders (a linear order can be seen as a tree, where the subtrees are the prefixes). A corollary of~\cite[Theorem 3.3]{linearcliquewidth2021} is that the main conjecture holds for every class of graphs of bounded linear rankwidth. Therefore, the consequences of the main conjecture will hold for classes of graphs of bounded linear rankwidth. 

Even in the linear case, however, much remains open: we do not know if the main conjecture is true for classes of structures that have linear rankwidth, but use a vocabulary with ternary relations (see~\cite[Example 9]{DBLP:journals/corr/abs-2305-18039} for some evidence on why ternary relations could be harder than binary ones). We return to bounded linear rankwidth in Appendix~\ref{sec:bounded-branching}. 

Below we discuss some consequences of the conjecture. 

\subsubsection*{Rank-decreasing equivalent to sub-definable.} The first consequence is the one that was stated at the beginning of this section: generalizing our main result, Theorem~\ref{thm:main-graphs}, to bounded rankwidth. 
Using the same proof as in Section~\ref{sec:endgame-bounded-treewidth}, the main conjecture would entail that for transductions that input structures of bounded rankwidth, being rank-decreasing is the same as being sub-definable. 

\begin{conjecture}[Implied by main conjecture] \label{conj:def-rank-decreasing} Consider a transduction whose inputs have bounded rankwidth. Then this transduction is rank-decreasing if and only if it is sub-definable. 
\end{conjecture}

As explained in Example~\ref{ex:grid-counterexample}, some assumptions on the input class are necessary for the conjecture. Although we have not investigated this enough to make conjectures, it could be the case that the assumption on bounded rankwidth is tight: maybe for every class $\classc$ of unbounded rankwidth there is a $\classc$-to-$\classc$ transduction that is rank-decreasing but not sub-definable. 

\subsubsection*{Recognisable languages.} One of the main themes of logic in computer science is the relation between recognisability (by automata, or algebras) and definability (in monadic second-order logic or its variants). A recent proposal on this topic was presented in~\cite[Section 5]{DBLP:journals/corr/abs-2305-18039}, which proposed a generic definition of recognisability of a language $L$ contained in a class of structures $\classc$. This proposed notion of recognisability  coincides with all known pre-existing notions of recognisability for specific classes of structures~\cite[Figure 1]{DBLP:journals/corr/abs-2305-18039}. We do not recall this notion here, but only the basic results about it:
\begin{enumerate}
 \item definability always implies recognisability (here and below, definability refers to definability in \cmso);
 \item if the class $\classc$ has unbounded rankwidth, then recognisability is not equivalent to definability;
 \item (conjectured) if the class $\classc$ has bounded rankwidth, then recognisability is equivalent to definability.
\end{enumerate}

The third item would be true, assuming a conjecture implied by the main conjecture of this paper. Putting the above observation together, we get the following 

\begin{conjecture}[Implied by main conjecture] 
 A class $\classc$ of structures has bounded rankwidth if and only if
 \begin{align*}
 \forall L \subseteq \classc \quad \text{$L$ is recognisable} \Leftrightarrow \text{$L$ is definable.}
 \end{align*}
\end{conjecture}

\subsubsection*{Two sub-goals.}
Given the importance of the main conjecture, it is worthwhile to split it into sub-conjectures. Here is one possible split. In the following, a transduction is called \emph{rank-invariant}  if it is rank-decreasing in both directions. In other words, the rank of a subset is asymptotically the same in the input and output structures.

\begin{conjecture}[Equivalent statement of main conjecture, see Fact~\ref{fact:conjectures-equivalent} below]\label{conj:more-refined}
 Let $\classc$ be a class of structures. 
 \begin{enumerate}
 \item\label{it:rank-preseving-trees} if $\classc$ has bounded rankwidth, then there is a surjective rank-invariant tree-to-$\classc$ transduction; and 
 \item\label{it:to-tree} if a $\classc$-to-tree transduction is rank-invariant, then it is sub-definable. 
 \end{enumerate}
\end{conjecture}

The assumption that $\classc$ has bounded rankwidth is not stated in item~\ref{it:to-tree} above, but it is implicit: if a transduction with tree outputs is rank-invariant, then its inputs have bounded rankwidth~\cite{courcelle1995logical}. In Appendix~\ref{sec:bounded-branching} we prove a special case of item~\ref{it:to-tree}, which corresponds to bounded linear rankwidth.

\begin{fact} \label{fact:conjectures-equivalent}
 Conjectures~\ref{conj:main} and~\ref{conj:more-refined} are equivalent.
\end{fact}
\begin{proof}
 Let us first show how the two items in Conjecture~\ref{conj:more-refined} imply Conjecture~\ref{conj:main}. Let $\beta$ be the rank-invariant transduction from item~\ref{it:rank-preseving-trees}. By Theorem~\ref{thm:from-trees} this transduction is sub-definable in the tree-to-$\classc$ direction, and by item~\ref{it:to-tree} it is sub-definable in the $\classc$-to-tree direction. 

 We now show the converse implication.
 Assume Conjecture~\ref{conj:main}. If a transduction is sub-definable in both directions, then it is rank-invariant by Corollary~\ref{cor:definable-is-rank-decreasing}. Therefore, the main conjecture gives item~\ref{it:rank-preseving-trees} in Conjecture~\ref{conj:more-refined}. Let us now prove item~\ref{it:to-tree}. Let $\beta$ be the tree-to-$\classc$ transduction from Conjecture~\ref{conj:main} that is sub-definable in both directions, and let $\alpha$ be a rank-invariant $\classc$-to-tree transduction. We want to show that $\alpha$ is sub-definable. Consider the following diagram:
 \[
 \begin{tikzcd}
 \classc
 \ar[d,"\alpha"]
 & 
 \text{trees}
 \ar[l,"\beta"]
 \ar[dl,"\beta; \alpha"]\\
 \text{trees}
 \end{tikzcd} 
 \]
 Since $\beta$ is sub-definable, it is rank-invariant, and therefore $\beta;\alpha$ is rank-invariant. Since  $\beta;\alpha$ inputs trees, it must be sub-definable by Theorem~\ref{thm:from-trees}. Finally, $\bar \beta; \beta; \alpha$ is sub-definable as a composition of the sub-definable transductions $\bar \beta$ and $\beta;\alpha$.  The transduction  $\alpha$ is a subset of this transduction, since  $\bar \beta; \beta$ contains the identity on $\classc$, and therefore $\alpha$ is sub-definable.
\end{proof}

%% file: appendix.tex
\section{Ranks of bounded quantifier depth}\label{app:ef-games}

We now prove Theorem~\ref{thm:rank-invariant-under-quantifier-rank}.
In the proof, it will be more convenient to use a variant of logic and the type matrix uses only set variables and does not use  first-order variables.  We begin by describing this variant.

    Define a \emph{monadic  second-order structure} to be a structure where all relations take only set arguments.  This means that a relation of arity $m$ describes an $m$-ary relation on  subsets of the universe.  (An example is a hypergraph, in which the edge relation is a family of subsets of the universe, or a matroid, which is a hypergraph subject to certain axioms. These two examples have only unary relations on sets, but it turns out that relations of higher arity do not make a big difference, see~\cite[Example 9]{DBLP:journals/corr/abs-2305-18039}.) For the purposes of the proof of Theorem~\ref{thm:rank-invariant-under-quantifier-rank}, we use the name \emph{first-order structure} for the usual kind of structures, where relations take only element arguments, to distinguish them from monadic second-order structures.  Every first-order structure $\structa$ can be converted into a monadic second-order structure, called its \emph{singleton lifting}, as follows:  replace every relation $R$ with a relation of  same arity (but taking set arguments) that is interpreted  as
    \begin{align*}
\setbuild{ (\set{a_1},\ldots,\set{a_n})}{$\structa \models R(a_1,\ldots,a_n)$}.
    \end{align*}

The notion of types is extended to monadic second-order structures in the usual way, by associating to a tuple of $m$ subsets in a given structure the set of all \mso formulas with $m$ free set-variables, of given quantifier depth, that hold for this tuple. The following observation shows that a structure is essentially the same as its singleton lifting.

    \begin{fact}\label{fact:monadic-lifting}
For every quantifier depth $d$ and number of arguments $m$, the following types are in one-to-one correspondence:  
    \begin{enumerate}
        \item  $d$-type of a first-order structure  $\structa$ with distinguished elements $a_1,\ldots,a_m$;
        \item $d$-type of the singleton lifting of $\structa$  with distinguished subsets $\set{a_1},\ldots,\set{a_m}$.
    \end{enumerate}
    \end{fact}


The matrices defining the rank are adapted to monadic second-order  structures in the obvious way, defined as follows. Consider a monadic second-order structure $\structa$, a  subset $X$ of its universe, a  quantifier depth $d$ and a number of arguments $m$.  Define 
    the following matrix, which we call the $d$-type matrix of arity $m$, for the subset $X$ in the structure $\structa$:
    \begin{enumerate}
    \item rows are $m$-tuples $(X_1,\ldots,X_m)$ of subsets of $X$;
    \item columns are $m$-tuples $(Y_1,\ldots,Y_m)$ of subsets of $\structa \setminus X$;
    \item the value of the matrix in a row $(X_1,\ldots,X_m)$ and  column $(Y_1,\ldots,Y_m)$ is the monadic $d$-type of the tuple
    $X_1 \cup Y_1,\ldots,X_m \cup Y_m$ in $\structa$.
\end{enumerate}

We denote this matrix by $M_{d,n}$, assuming that the structure $\structa$ and subset $X$ are implicit from the context. The following lemma  rephrases Theorem~\ref{thm:rank-invariant-under-quantifier-rank} in terms of monadic second-order structures.

\begin{lemma}\label{lem:rank-with-quantifier-depth-d} Fix  a class of monadic second-order structures, where $m$ is the maximal arity of relations. For every quantifier depth $d \in \set{0,1,\ldots}$, the following  rank functions, each of which inputs a structure $\structa$ with a distinguished subset of elements $X$, are asymptotically equivalent: 
    \begin{enumerate}
        \item \label{it:0m} number of distinct rows in the matrix $M_{0,m}$;
        \item \label{it:0d+m} number of distinct rows in the matrix $M_{0,d+m}$;
        \item \label{it:0dm} number of distinct rows in the matrix $M_{d,m}$.
    \end{enumerate}
\end{lemma}

By Fact~\ref{fact:monadic-lifting}, the equivalence between items~\ref{it:0m} and~\ref{it:0dm} in Lemma~\ref{lem:rank-with-quantifier-depth-d} implies Theorem~\ref{thm:rank-invariant-under-quantifier-rank}. Therefore, it remains  to prove the lemma.

We begin with the equivalence of the first two items. 
\begin{lemma}
    The rank functions in items~\ref{it:0m} and~\ref{it:0d+m} of Lemma~\ref{lem:rank-with-quantifier-depth-d} are asymptotically equivalent.
\end{lemma}
\begin{proof}
    Since $m$ is the maximal arity of relations in the vocabulary,  a quantifier-free type is uniquely determined by its projections onto $m$-tuples of variables. 
\end{proof}

We are left with showing the equivalence of the rank function in item~\ref{it:0dm} of Lemma~\ref{lem:rank-with-quantifier-depth-d}. Clearly the matrix $M_{d,m}$  corresponding to item~\ref{it:0dm} stores more information than the matrix $M_{0,m}$ corresponding to item~\ref{it:0m}, and therefore the corresponding rank function can only be bigger. The following lemma shows that it cannot be uncontrollably bigger, because it is bounded by a function of the rank of the  matrix $M_{0,d+m}$ corresponding to item~\ref{it:0d+m}.
\begin{lemma}\label{lem:exponential-bound}
    For every monadic second-order  structure  $\structa$, every  subset $X$ of its universe, and every $d,m \in \set{0,1,\ldots}$, the number of distinct rows in the matrix $M_{d+1,m}$ is at most 
    \begin{align*}
        2^{\text{number of distinct rows in the matrix $M_{d,m+1}$}}.
    \end{align*}
\end{lemma}
By applying this lemma $d$ times, we get the equivalence of items~\ref{it:0d+m} and~\ref{it:0dm}, although the bound uses a tower of $d$ exponentials. (This is not particularly surprising, since a tower of exponentials is the usual dependency of the number of types on the  quantifier depth.) Therefore, the lemma  concludes the proof of Theorem~\ref{thm:rank-invariant-under-quantifier-rank}, and it only  remains to prove it. 
\begin{proof}[Proof of Lemma~\ref{lem:exponential-bound}]  This is a symbol-pushing proof. We begin by introducing some notation and assigning types to the objects that will be used. 
    
    In the lemma, we are counting \emph{distinct} rows in a matrix, so we use the following notation: a \emph{row index} is the name of a row, while a \emph{row content} is the sequence of values in the row, which is indexed by columns in the matrix.
    In the   matrix  $M_{d+1,m}$, a row index is   a choice of subsets $X_1,\ldots,X_m \subseteq X$, and the corresponding row content, which we denote by  
    \begin{align*}
    M_{d,m}[(X_1,\ldots,X_m), \_]
    \end{align*}
    is a function of type
    \begin{align*}
    \myunderbrace{(\powerset(\structa \setminus X))^m}{columns of the matrix}
    \to 
    \myunderbrace{\text{$(d+1)$-types with $m$ variables}}{type of values in the matrix}.
    \end{align*}

We want to estimate the number of distinct rows, i.e.~the number of possible row contents, in the matrix $M_{d+1,m}$, based on the corresponding number for  the matrix $M_{d,m+1}$. This is done in the following claim, which immediately implies the lemma.
\begin{claim}
    In the matrix $M_{d+1,m}$, the row contents 
    \begin{align*}
    M_{d+1,m}[(X_1,\ldots,X_m), \_] 
    \end{align*}
are uniquely determined by the following set of row contents
    \begin{align*}
        \setbuild
        {M_{d,m+1}[(X_1,\ldots,X_m,X_{m+1}), \_] }
        {$X_{m+1} \subseteq X$}.
        \end{align*}
\end{claim}
\begin{proof}
    This is a standard  Ehrenfeucht-Fraïssé argument.
    Consider two rows in the matrix $M_{d+1,m}$ 
    \begin{align*}
    X_1,\ldots,X_m \subseteq X \\
    X'_1,\ldots,X'_m \subseteq X 
    \end{align*}
    that agree on the information in the statement of the claim, which means that the following two sets of row contents are the same 
    \begin{align*}
        \setbuild
        {M_{d,m+1}[(X_1,\ldots,X_m,X_{m+1}), \_] }
        {$X_{m+1} \subseteq X$}
        \\
        \setbuild
        {M_{d,m+1}[(X'_1,\ldots,X'_m,X'_{m+1}), \_] }
        {$X'_{m+1} \subseteq X$}.
        \end{align*}
        We need to show that the row contents 
        \begin{align*}
            M_{d+1,m}[(X_1,\ldots,X_m),\_] \\
            M_{d+1,m}[(X'_1,\ldots,X'_m),\_]
            \end{align*}
        are equal. By unravelling the definition of the matrix $M_{d+1,m}$, this means that for every column 
        \begin{align*}
        Y_1,\ldots,Y_m \subseteq \structa \setminus X,
        \end{align*}
        player Duplicator can win the $d+1$ round Ehrenfeucht-Fraisse game on the structure $\structa$ with the two sides being 
        \begin{align*}
            X_1 \cup Y_1,\ldots,X_m  \cup Y_m  \\
            X'_1 \cup Y_1 ,\ldots,X'_m \cup Y_m . 
            \end{align*}
        Consider a move by player Spoiler in this game. By symmetry, we assume that Spoiler plays in the first side, and he chooses a subset $X_{m+1} \cup Y_{m+1}$.  By the assumption on the sets of row contents in $M_{d,m+1}$ being equal,  there is some $X'_m$ such that the following row contents are equal
        \begin{align*}
            M_{d,m+1}[(X_1,\ldots,X_m,X_{m+1}),\_] \\
            M_{d,m+1}[(X'_1,\ldots,X'_m,X'_{m+1}),\_].
        \end{align*}
        This implies that the matrix $M_{d,m+1}$ has the same values in the entries corresponding to 
        \begin{align*}
            X_1 \cup Y_1,\ldots,X_{m+1}  \cup Y_{m+1}  \\
            X'_1 \cup Y_1 ,\ldots,X'_{m+1}\cup Y_{m+1} 
            \end{align*}
        and therefore  player  Duplicator can use $X'_{m+1} \cup Y_{m+1}$ as her response to win the game. 
\end{proof}
This completes the proof of Theorem~\ref{thm:rank-invariant-under-quantifier-rank}.
\end{proof}

%% file: linear-case.tex
\section{Bounded branching}
\label{sec:bounded-branching}
Recall item~\ref{it:to-tree} of Conjecture~\ref{conj:more-refined}, which said that if a $\classc$-to-tree transduction is rank-invariant, then it is sub-definable. In this section, we prove a special case of the conjecture, which is the most technically advanced contribution of the paper. This result will work with structures that are not necessarily graphs, i.e.~we allow relations in the vocabulary that have arity bigger than two, and therefore we cannot use the results from~\cite{linearcliquewidth2021}. 

We begin by observing that the assumption on being rank-invariant cannot be weakened to rank-decreasing. 

\begin{example}\label{ex:grid-to-tree-transduction}
    Here is a transduction that outputs trees, is rank-invariant, but is not sub-definable. Consider the transduction which inputs a grid, and outputs all possible trees over the universe of the input grid. Like every transduction that inputs grids, this transduction is rank-decreasing. However, it is not sub-definable, since the number of trees on a universe of size $n$ is bigger than the number of possible outputs for a sub-definable transduction. This is the because the former is $2^{\Theta(n \cdot \log n)}$ while the latter is $2^{\Theta(n)}$.
\end{example}

Define a \emph{minor} of a tree to be a tree that is obtained by removing some leaves, and restricting the family of subtrees to the smaller set of leaves. 
The \emph{branching} of a tree is defined to be the largest $k$ such that the tree contains, as a minor, the complete binary tree of height $k$.  We will be interested in classes of trees of bounded branching. An equivalent description of bounded branching is that some tree is avoided as a minor (this is because every tree is a minor of some complete binary tree). 

The goal of this section is to prove the following special case of item~\ref{it:to-tree} in Conjecture~\ref{conj:more-refined}, in which the outputs have bounded branching. 
\begin{theorem}\label{thm:to-tree-linear}
    Let $\classc$ be a class of structures. If a $\classc$-to-tree  transduction is rank-invariant, and has outputs of bounded branching, then it is sub-definable.
\end{theorem}

Note that the transduction in the above theorem is also sub-definable in the converse direction, by Theorem~\ref{thm:from-trees}. The rest of this section is devoted to proving Theorem~\ref{thm:to-tree-linear}.



\subsubsection{Partially ordered trees}
In the proof of the above theorem, we will be working with a notion that is equivalent to trees of bounded branching, namely partially ordered trees of bounded height. These are trees where some, but not all, nodes  are equipped with a linear order on their children.

\begin{definition}[Partially ordered trees]
    A \emph{partially ordered tree}  is a tree together with the following extra structure:
    \begin{enumerate}
        \item a partition of the nodes into two kinds:  ordered and unordered nodes;
        \item for each ordered node, a linear order on its children. 
    \end{enumerate} 
\end{definition}

To represent partially ordered trees as logical structures, we rely on the following notion.
Define the \emph{document order} of a partially ordered tree to be the following preorder on leaves: $x \le y$ if either $x=y$, or otherwise  the closest common ancestor of $x$ and $y$ is an ordered node, and the child containing $x$ is to the left of the child containing $y$. If the closest common ancestor of the two leaves is an unordered node, then they are incomparable in the document order. In particular, if a tree has only unordered nodes, then  every two nodes are incomparable.  It is not hard to see that a partially ordered tree is uniquely specified by the family of nodes and its document order.

The following lemma shows that partially ordered trees of bounded height are the same as trees of bounded branching. 
\begin{lemma}\label{lem:branching-po-trees}
    For every $k$ there is a transduction between:
    \begin{enumerate}
        \item trees of branching $\le k$; and 
        \item partially ordered trees of height $\le 2k$,
    \end{enumerate}
    which is bijective and  definable in both directions. 
\end{lemma}

Thanks to Lemma~\ref{lem:branching-po-trees}, an equivalent statement of Theorem~\ref{thm:to-tree-linear} is: if a transduction that outputs partially ordered trees of bounded height is  rank-invariant, then it is sub-definable. We will  prove this induction on the height of the output trees, and in the induction step, we will show that the local tree structure is sub-definable.

\subsubsection{Local tree structure}
We begin by formalising the concept of local tree structure in a partially ordered tree. 

\begin{definition}
    Let $X$ be a node  in a partially ordered tree. 
\begin{enumerate}
    \item  Define the  \emph{child partition} to  be the equivalence relation on $X$ which identifies two elements if they are in the same child.
    \item Assume that $X$ is an ordered node. Define the \emph{child preorder} to be the linear preorder on $X$ which is obtained by ordering the equivalence classes of the child partition according to the order on children of $X$.
\end{enumerate} 
\end{definition}

The child preorder is represented as a linear preorder, i.e.~a relation that is reflexive, transitive, and total (i.e.~every two elements are related in at least one direction). The following lemma will be used several times in this section to show sub-definability for transductions that output partially ordered trees.

\begin{lemma}\label{lem:sufficient-condition}
    Let $\alpha$ be a transduction that outputs partially ordered trees of bounded height.  A sufficient condition for $\alpha$ being sub-definable is that the following two transductions are sub-definable: 
         \begin{align*} 
            \myunderbrace{
                \setbuild{(\structa|X, \sim)}{$X$ is an unordered node for some  $\structt \in \alpha(\structa)$ \\ and $\sim$ is the  corresponding child partition }
            }{ this transduction is called the local unordered transduction of $\alpha$} 
            \\
            \myunderbrace{
            \setbuild{(\structa|X, \le)}{$X$ is an ordered node for some  $\structt \in \alpha(\structa)$ \\ and $\le$ is the  corresponding child preorder } }{
                this transduction is called the local unordered transduction of $\alpha$
            }.   \end{align*}
\end{lemma}
\begin{proof}
    Induction on the maximal height of the output trees. Define $\beta$ to be the transduction 
    \begin{align*}
    \setbuild{(\structa|X, \structt|X)}{$(\structa,\structt) \in \alpha$ and $X$ is a child subtree in $\structt$}.
    \end{align*}
    To this transduction, we can apply the induction assumption, and therefore it is sub-definable. The original transduction is then obtained from by combining $\beta$ with the local transduction, in the special case where  $X$ is the root of the tree.
\end{proof}

In  Sections~\ref{sec:root-is-unordered} and~\ref{sec:root-is-ordered} below,  we will show that for every rank-invariant transduction that outputs partially ordered trees, the two local transductions from the above lemma are sub-definable. This will imply that if, furthermore, the output trees have bounded height, then the transduction is sub-definable. This will, in turn, imply Theorem~\ref{thm:to-tree-linear}, due to the isomorphism between partially ordered trees of bounded height and trees of bounded branching that was stated in Lemma~\ref{lem:branching-po-trees}.

    \subsection{Unordered local transduction}
    \label{sec:root-is-unordered}

    The goal of this section is to prove the following lemma.

    \begin{lemma}\label{lem:unordered-case}
        For every rank-invariant transduction that outputs partially ordered trees, its local unordered transduction is sub-definable.
    \end{lemma}

    The rest of Section~\ref{sec:root-is-unordered} is devoted to proving the above lemma. The assumption on rank-invariance cannot be relaxed to being rank-decreasing, for the same reasons as in Example~\ref{ex:grid-to-tree-transduction}. 
    The proof of the lemma proceeds in two steps. First, we show  that the local unordered transduction satisfies a certain technical condition, and then we show that every transduction with this condition is sub-definable.

\subsubsection{Unordered approximation}
\label{sec:unordered-approximation}
In this section, we define and prove the technical condition that will later be used to prove sub-definability.

The general idea is that we can approximate the output equivalence relation using logic in the input structure. What do we mean by approximating an equivalence relation?
Recall Example~\ref{ex:equivalence-relations}, in which we showed that in a structure which is an equivalence relation, the rank of a subset is approximately the same as the number of equivalence classes that are cut by the subset. In the following lemma, we will show that  we can define in logic for every subset an approximation of the  number of equivalence classes that are cut by this subset. The lemma also contains an additional recognisability condition; as far as we know the remainder of the proof could be improved so that the recognisability condition is not needed. 

    \begin{lemma}\label{lem:technical-condition-unordered} Let $\alpha$ be a rank-invariant transduction that outputs partially ordered trees, and let  $\beta$ be its local unordered transduction. Then $\beta$  satisfies the following condition, which we call \emph{unordered approximation}.
        There is a \cmso formula $\varphi$ over the input vocabulary, a finite commutative semigroup $S$ and  some $k \in \set{1,2,\ldots}$ such that every pair  $(\structa, \sim) \in \beta$  satisfies all the following conditions:
        \begin{enumerate}
            \item \emph{Completeness.} If  a subset of $\structa$  cuts no equivalence classes, then it satisfies $\varphi$;
            \item \emph{Approximate soundness.}   If a subset of $\structa$ cuts at least $k$ equivalence classes, then it does not satisfy $\varphi$;

            \item \emph{Recognisability.}  Let $X_1,\ldots,X_n$ be the equivalence class of $\sim$. There is a function\footnote{In the type of the function $\lambda$ we use disjoint union, since the empty subset appears in each of the powersets $\powerset X_1, \ldots, \powerset X_n$, and we want $\lambda$ to be able to assign different values to the different instances of the empty subset.}
            \begin{align*}
                \lambda: \powerset X_1 \uplus \cdots \uplus \powerset X_n  \to S
            \end{align*}
            with the following property. Let $Y \subseteq \structa$, and let $s_1,\ldots,s_n$ be the semigroup elements where $s_i$ is obtained by applying $\lambda$ to the set $Y \cap X_i$. Then  the product $
            s_1 \cdots s_n$ in the semigroup uniquely determines if the set $Y$ satisfies $\varphi$. Observe that since the semigroup is commutative, the order of arguments in the product is not important.
            \item The value of $\lambda$ determines if a subset is full or empty.
        \end{enumerate}
    \end{lemma}


    \begin{proof}
        Consider a pair  $(\structa, \structt) \in \alpha$, and an unordered node $X$ in the tree $\structt$. If a set  $Y \subseteq X$ does not cut any children of $X$,  then its rank  in $\structt$ is bounded. This, together with the assumption that   $\alpha$ is rank-invariant,  yields some constant $\ell \in \set{1,2,\ldots}$, such that 
        \begin{enumerate}
            \item[(a)] for every $(\structa,\structt) \in \alpha$, every unordered node $X$ of $\structt$ and every  $Y \subseteq X$, if $Y$ does not cut any children of $X$, then the rank of $Y$ in $\structa |X$ is  at most $\ell$ (Lemma~\ref{lem:rank_in_substructure}).
        \end{enumerate}
        Using again the assumption that the transduction $\alpha$  is rank-invariant, we know that  sets $Y \subseteq X$ which have rank  at most $\ell$ in $\structa|X$ will have bounded rank in $\structt|X$. 
        A set of bounded rank in a partially ordered tree can cut a bounded number of children of any given node, and therefore there is some constant $k$ such that 
        \begin{enumerate}
            \item[(b)] for every $(\structa,\structt) \in \alpha$, every unordered node $X$ of $\structt$ and every  $Y \subseteq X$,  if $Y$ has rank at most $\ell$ in $\structa|X$ then it cuts at most  $ k$  children of $X$.
        \end{enumerate}

        Define $\varphi(Y)$ to be the formula which says that the set $Y$ has rank at most $\ell$ in the structure $\structa|X$.  This is indeed a formula of \mso, even without counting, since an \mso formula can count distinct rows in a type matrix up to a given threshold. 
        Conditions (a) and (b) give the completeness and approximate soundness  conditions from the lemma. 

        It remains to prove the recognisability condition in the  lemma. Since the transduction $\alpha$ is rank-invariant, it follows that its converse is rank-decreasing. Since the converse inputs trees, it is sub-definable by Theorem~\ref{thm:from-trees}. By pulling back the formula $\varphi$ along the converse, there is a \cmso formula 
        \begin{align*}
            \psi(Y,Z_1,\ldots,Z_n)
        \end{align*}
        over the vocabulary of partially ordered trees, such that for every pair $(\structa,\structt) \in \alpha$ there is some choice of parameters $Z_1,\ldots,Z_n \subseteq \structt$ such that 
        \begin{align*}
        \structa \models \varphi(Y)
        \quad \text{iff} \quad 
        \structt \models \psi(Y,Z_1,\ldots,Z_n)
        \end{align*}
        holds for all subsets $Y \subseteq \structa$.
        Let $d$ be the quantifier depth of the formula $\psi$. 
        By compositionality of \mso on trees, we know that if $X$ is a node of the tree $\structt$ and $Y \subseteq X$, then whether the right-hand side of the above equivalence holds depends only on the values
        \begin{align}\label{eq:colour-homogeneity}
            \text{$d$-type of $Y$ in  $(\structt, Z_1,\ldots,Z_n) | X_i$},         \end{align}
            where $X_i$ ranges over children of $X$. 
        Furthermore, the dependence is commutative, i.e.~it is mediated  by a homomorphism into  a finite commutative semigroup, namely the semigroup of $d$-types with respect to disjoint unions of structures. This yields the commutative semigroup $S$ from the homogeneity condition in the lemma, with the colouring $\lambda$ mapping a subset $Y \subseteq X_i$ to the $d$-type from~\eqref{eq:colour-homogeneity}.
    \end{proof}

\subsubsection{Sub-definability using unordered approximation}
\label{sec:homogeneous-unordered-approximation} We now show that the unordered approximation condition, as defined in Lemma~\ref{lem:technical-condition-unordered}, implies sub-definability. In the proof, we use a stronger version of this condition, defined as follows. We say that a transduction has the \emph{homogeneous unordered approximation property} if it satisfies the unordered approximation property (see Lemma~\ref{lem:technical-condition-unordered}), together with the following extra condition: 
\begin{enumerate}
    \item[(4)] Homogeneity. All the images are equal:
    \begin{align*}
    \lambda(\powerset X_1) = \cdots = \lambda(\powerset X_n).
    \end{align*}
\end{enumerate}
The following lemma shows that this stronger property is sufficient for sub-definability.

    \begin{lemma}\label{lem:use-unordered-approximation} If a transduction has the homogeneous unordered approximation property, then it  is sub-definable.
\end{lemma}
    \begin{proof}
    Let $\beta$ be a transduction with the homogeneous unordered approximation property. Let the semigroup $S$, the threshold $k$ and the formula $\varphi$ be as in this property. 
    
    Consider a pair $(\structa,\sim) \in \beta$, let $X_1,\ldots,X_n$ be the equivalence classes of $\sim$, and let 
    \begin{align*}
    \lambda : \powerset X_1 \uplus \cdots \uplus \powerset X_n \to S
    \end{align*}
    be the function from the unordered approximation property. 
 For each class $X_i$, we will be especially interested in two values produced by the function $\lambda$:  the \emph{empty value}, which is produced for the empty set $\emptyset \subseteq X_i$ and the \emph{full value}, which is produced for the full subset $X_i \subseteq X_i$.

    \begin{claim}\label{claim:same-full-empty-values}
         Without loss of generality, we can assume that all classes have the same empty value, all classes have the same full value, and these two values are idempotent.
    \end{claim}
    \begin{proof}
        Without loss of generality, we can assume that the colouring $\lambda$ in the recognisability condition also keeps track of whether a set is full or empty, which is an idempotent property. (In fact, the proof of Lemma~\ref{lem:technical-condition-unordered} already ensures this, since $\lambda$ returns the \mso type of sufficiently high quantifier depth, and such a type will tell us if a set is empty or full, as long as the quantifier depth is at least one.)
        From the homogeneity condition, we know that $\lambda$ has the same image on the powerset of every class, and hence all full values (respectively, all empty values) must be the same. 
    \end{proof}
    
    We say that a set $Y \subseteq \structa$ is a \emph{seed} if it satisfies all the following conditions:
    \begin{enumerate}
        \item \label{it:maximal-commutative-phi} $Y$ satisfies the formula $\varphi$;
        \item \label{it:maximal-commutative-full} $Y$ is full in at  least one equivalence classes;
        \item \label{it:maximal-commutative-empty}  $Y$ is empty in  at least one equivalence classes.
    \end{enumerate}

    Note that as long as there are at least two classes, there is a seed: simply take any set satisfying (2) and (3) that does not cut any class, then it satisfies (1) by the completeness property.

    Choose a seed $Y_0$ that cuts a maximal number of equivalence classes. Note that $Y_0$ cuts at most $k$ classes by approximate soundness.
    Fix an arbitrary class for which $Y_0$ is full, and an arbitrary class for which $Y_0$ is empty, then define a \emph{special equivalence} class to be either these two fixed classes, or any class that is cut by $Y_0$, as in the following figure.
    \mypic{24}

 The main observation is in the following claim, which shows that the equivalence classes can be determined by looking at sets that are similar to the seed $Y_0$. 
    
    \begin{claim}\label{claim:seed-separator} Assume that there are at least $2$ equivalence classes.
        Consider two elements of the universe that are not in any special class.
        Then they are in different classes if and only if they can be separated by some seed that agrees  with $Y_0$ on the special classes. 
    \end{claim}
    \begin{proof}
        By maximality, if a seed agrees with $Y_0$ on the special classes, then it cannot cut any class that is not special. Therefore, such a seed cannot separate two elements in the same  class. 
        
        We now prove the converse implication. As explained above,  if a seed agrees with $Y_0$ on the special classes, then it is full or empty on every non-special class. Since the values of the full and empty set are idempotent, and the semigroup is commutative, the set which agrees with $Y_0$ on the special classes, is full on the class containing one of the elements, and empty elsewhere, is a seed separating the two elements.
    \end{proof}

    We are now ready to prove that the transduction is sub-definable. The case where there is a single class can be explicitly guessed by the transduction.
    Otherwise, the transduction can guess the seed and the special equivalence classes, and then use the condition in the above claim to check which elements are in the same equivalence class.
\end{proof}
   
\subsubsection{Ensuring homogeneity}
\label{sec:unordered-ensuring-homo}
In Lemma~\ref{lem:technical-condition-unordered}, we have shown that the unordered local transduction has the unordered approximation property, and in Lemma~\ref{lem:use-unordered-approximation} we have shown that the stronger homogeneous version of this property is sufficient for sub-definability. We finish the proof by bridging the gap between these two properties. The idea is that every transduction with the unordered approximation property can be refined to ensure the homogeneous version. 

This refinement will be relatively straightforward in the present unordered case, but it will be more involved (using factorisation forests) in the ordered case. To cover both cases, we present below a general definition of what we mean by refining a transduction.

    \begin{definition}[Refinement on trees] A partial tree $\structs$ is called \emph{finer} than a partial tree $\structt$ if 
    \begin{enumerate}
        \item the two trees have the same universe, i.e.~leaves;
        \item the two trees have the same document preorder;
        \item every node of $\structs$ is also a node of  $\structt$.
    \end{enumerate}
  A transduction $\beta$ is \emph{refined by} a transduction $\gamma$ if both transductions output partially ordered trees, and for every input structure $\structa$, every output tree in $\beta(\structa)$ is refined by some output tree in $\gamma(\structa)$.
    \end{definition}

    We have now all the definitions in place that are needed to complete the proof of Lemma~\ref{lem:unordered-case}.  

    \begin{proof}[Proof of Lemma~\ref{lem:unordered-case}] 
        Let $\beta$ be a transduction as in the assumptions of the lemma, i.e.~this is the local unordered transduction of some rank-invariant transduction that outputs partially ordered trees. We want to show that $\beta$ is sub-definable.
        
        By Lemma~\ref{lem:technical-condition-unordered}, $\beta$ has the unordered approximation property. 
        We can view of $\beta$ as a transduction that outputs trees, since an  equivalence relation can be seen as a special case of a partially ordered tree, which has height two, as in the following picture:
        \mypic{22}
        It is this view that we use when talking about a refinement of an equivalence relation, or a transduction that outputs equivalence relations.

        \begin{claim}\label{claim:refine-unordered}
            Let $\beta$ be a transduction that outputs equivalence relations, and has the unordered approximation property.
            There is a transduction $\gamma$ that refines $\beta$, has outputs of bounded height, and such that the local unordered transduction of $\gamma$ has the homogeneous unordered approximation property. 
        \end{claim}
        \begin{proof}
            Group the equivalence classes  into a bounded number of groups, with two classes being together in a group when they have the same image of the function $\lambda$. This grouping can be seen as tree of height three, in which the root has a bounded number of children, since the semigroup is fixed. For every group, its children satisfy the homogeneity condition. 
        \end{proof}

        The proof of the above claim is very straightforward. The reason why we phrase the claim in terms of refinements is that we want  to underline the similarities with the ordered case in the next section, where the analogue of Claim~\ref{claim:refine-unordered} will be more advanced.

        Consider the transduction $\gamma$ from Claim~\ref{claim:refine-unordered}. This transduction has outputs of bounded height. Its local unordered transduction is sub-definable by the claim, and its local ordered transduction is sub-definable for vacuous reasons: all output trees of $\gamma$ use only unordered nodes, since the same is true for $\beta$. Therefore, $\gamma$ is sub-definable thanks to Lemma~\ref{lem:sufficient-condition}. Therefore, $\beta$ is also sub-definable, since it can be obtained from $\gamma$ by forgetting the extra levels with the groups of equivalence classes.
    \end{proof}


\subsection{The ordered local transduction}
\label{sec:root-is-ordered}

In this section, we prove the ordered analogue of Lemma~\ref{lem:unordered-case}, stated below.

\begin{lemma}\label{lem:ordered-case}
    For every rank-invariant transduction that outputs partially ordered trees, its local ordered transduction is sub-definable.
\end{lemma}

The proof follows the same structure as in  Section~\ref{sec:root-is-unordered}.

\subsubsection{Ordered approximation}
\label{sec:ordered-approximation}
We begin with the first step, in which we show that the local ordered transduction has a certain approximation property. In the unordered case, where the outputs were equivalence relations, the approximation property counted the number of equivalence classes that were cut. In the ordered case, we will be counting blocks, as explained below. 

Consider a linear preorder $\le$. In the following, a \emph{class} of this preorder is defined to be an equivalence class of the accompanying equivalence relation.  
An \emph{interval} in the preorder $\le$ is defined  to be any subset which has the property that if $x$ and $y$ are in the interval, then the same is true for any element $z$ with $x \le z \le y$. In particular, an interval cannot  cut any class. For a subset $Y$ in a linear preorder, define a \emph{block} to be an interval that has one of the following three kinds: 
\begin{enumerate}
    \item a maximal inclusion-wise interval contained in $Y$;
    \item a maximal inclusion-wise interval disjoint with $Y$;
    \item a class that is cut by $Y$.
\end{enumerate}
We will use the name ``full block'', ``empty block'' and ``cut block'' for the three kinds, respectively. Here is a picture:
    \mypic{18}
Observe that the blocks partition the equivalence classes induced by the preorder.

The following lemma, whose straightforward proof is left to the reader, shows that the number of blocks is approximately equal to the rank of a subset.
\begin{lemma}\label{lem:cut-for-orders}
    In the class of linear preorders, the rank of a subset  is asymptotically equivalent to the number of blocks.
\end{lemma}

We are now ready to state the ordered version of the approximation property. The proof of the following lemma is the same as in the unordered version.

\begin{lemma}\label{lem:technical-condition-ordered}  Let $\alpha$ be a rank-invariant transduction that outputs partially ordered trees, and let  $\beta$ be its local ordered transduction. Then $\beta$  satisfies the following condition, which we call \emph{ordered approximation}.
    There is a \cmso formula $\varphi$ over the input vocabulary,   a finite (not necessarily commutative) semigroup $S$ and  some $k \in \set{1,2,\ldots}$ such that every pair  $(\structa, \le) \in \beta$ satisfies all the following conditions:
    \begin{enumerate}
        \item \emph{Completeness.} If a subset of $\structa$ is an interval, then it satisfies $\varphi$;
        \item \emph{Approximate soundness.}   If a subset of  $\structa$ has at least $k$ blocks, then it does not satisfy $\varphi$;

        \item \emph{Recognisability.}  Let $X_1,\ldots,X_n$ be the equivalence classes of $\le$, listed according to the order. There is a function
        \begin{align*}
            \lambda: \powerset X_1 \uplus \cdots \uplus \powerset X_n  \to S
        \end{align*}
        with the following property. Let $Y \subseteq \structa$, and let $s_1,\ldots,s_n$ be the semigroup elements where $s_i$ is obtained by applying $\lambda$ to the set $Y \cap X_i$. Then  the product $
        s_1 \cdots s_n$ in the semigroup uniquely determines if the set $Y$ satisfies $\varphi$.\end{enumerate}
\end{lemma}

\subsubsection{Sub-definability using ordered approximation}
\label{sec:homogeneous-ordered-approximation} 
We now show how the ordered variant of the approximation condition can be used to prove sub-definability. As in the unordered variant, we define a homogeneous strengthening of the approximation property.

We say that a transduction which outputs linear preorders has the \emph{homogeneous ordered approximation property} if it satisfies the ordered approximation property (see Lemma~\ref{lem:technical-condition-ordered}), together with the following extra condition: 
\begin{enumerate}
    \item[(4)] Homogeneity. All the images are equal:
    \begin{align*}
    \lambda(\powerset X_1) = \cdots = \lambda(\powerset X_n),
    \end{align*}
    and furthermore the above image is an idempotent in the powerset semigroup $\powerset S$.
\end{enumerate}
Later, in Section~\ref{sec:ordered-ensuring-homo}, we will use the Factorisation Forest Theorem to ensure the above strengthening. The following lemma shows how it can be used to prove  sub-definability.

    \begin{lemma}\label{lem:use-ordered-approximation} If a transduction has the homogeneous ordered approximation property, then it  is sub-definable.
\end{lemma}
\begin{proof}
Let the formula $\varphi$, the semigroup $S$ and the constant $k$ be as in the ordered approximation property.

    Consider a pair $(\structa,\le)$ in the transduction, with the classes being $X_1,\ldots,X_n$, listed according to the order, and let $\lambda$ be the function from the ordered approximation property. Define the full and empty values for each class in the same way as in the unordered case. Using the same argument as in the unordered case, see Claim~\ref{claim:same-full-empty-values}, we can assume that all the full values are the same, all the empty values are the same, and these two values are idempotents. 

    Similarly to the unordered case, the crucial notion will be that of \emph{seed}. The ordered version of seed is defined just as before: it is a subset $Y \subseteq \structa$ that  satisfies $\varphi$, is full in at least one class, and is empty in at least one class. If a set is a seed, then it has at most $k$ blocks, like any set that satisfies $\varphi$, by the approximate soundness condition.

    Define a \emph{long block} to be a block that has at least two classes. Define a \emph{maximal seed} to be one that has a maximal number of blocks that are full or empty among all seeds.
    \begin{claim}\label{claim:good-seed}
        Assume that there are at least two classes.
        One can choose a maximal seed $Y_0$ so that it has exactly two long blocks, one of them full, one of them empty, and between these two blocks all other blocks are cut blocks, as in the following picture:
        \mypic{19}
    \end{claim}
    \begin{proof} Blocks that are full or empty can be stretched or compressed without affecting the image under $\lambda$, thanks to the homogeneity assumption and the fact that the full and empty values are idempotent.
    Pick a seed $Y_0$ and pick two blocks, one of them full, one of them empty, such that there are only cut-blocks between them. Then stretch them until all other full or empty blocks are reduced to size one. 
    \end{proof}

    Consider the seed $Y_0$ from the above claim. By symmetry, we assume from now on that the left block is full, and the right long block is empty. Define the \emph{special classes} to be all  the following classes: up to and including the first class in the left long block, and symmetrically, from the last class in the right long block until the end.  Here is a picture: 
    \mypic{20}
    Since the number of special classes is at most  the number of blocks, there is a bounded number of special classes.  

    Define a \emph{good seed} to be one that agrees with $Y_0$ on the special classes. By maximality, a good seed has at most two long blocks, which are separated by several cut blocks (we call these cut blocks the \emph{island}), as in the following picture: 
    \mypic{21}
    Furthermore, the island can be moved around at will, with the two long blocks stretching to accommodate it, by the assumption on homogeneity. Since the island has at most $k$ classes, we get the following claim. 

    \begin{claim}\label{claim:order-on-elements-with-good-colour-distance}
        Consider two elements $x$ and $y$ that are  in non-special classes, and which are separated by at least $k$ classes. Then  $x < y$ if and only if 
         there is some good seed $Y$ that  contains $x$ but not $y$. 
    \end{claim}

In the claim above, we use the assumption that the two elements are separated by at least $k$ classes. If we do not know the order, then it is difficult to check this assumption. To work around this problem, we will work with a stronger condition,  which can be checked by a transduction. This condition will use modulo counting, as described below.
    Define the \emph{index} of an element of the universe to be the number of its class in the order of equivalence classes. The index belongs to $\set{1,\ldots,n}$, where $n$ is the number of classes. A sufficient condition for having distance at least $d$ is that the indices differ by $d$ modulo $2d$. Therefore, a corollary of Claim~\ref{claim:order-on-elements-with-good-colour-distance} is the following claim.
    
    \begin{claim}\label{claim:colour-distance}  Let $d \ge k$.
        Consider two elements $x$ and $y$ that are in non-special classes, and whose indices differ by $d$ modulo $2d$.  Then  $x < y$ if and only if there is some good seed $Y$ that  contains $x$ but not $y$. 
    \end{claim}
    
    In Claims~\ref{claim:good-seed}--\ref{claim:colour-distance} we were working with a fixed pair $(\structa,\le)$, and the notions of seed, special class, good seed,  and colour were relative to that pair. 
    We now put together the observations so far, with the pair $(\structa,\le)$ no longer being fixed, to show that the relation from the previous claim can be produced as outputs  by a sub-definable transduction.
    \begin{claim}\label{claim:colour-distance-transduction} Let $d \ge k$.  There is  a sub-definable transduction $\gamma$, such that for every $(\structa,\le) \in \beta$, one of the outputs in $\gamma(\structa)$ is the binary relation 
\begin{align*}
    \setbuild{(x,y) \in \structa^2}{$x < y$ and the indices differ  by $d$  modulo $2d$}.
\end{align*}
    \end{claim}
    \begin{proof}
        The seed $Y_0$ and the special classes can be guessed by the transduction.  The remaining notions are definable in \mso.
    \end{proof}

    From the relation in the output of the above transduction, we can recover the $d$-fold successor relation defined by
    \begin{align*}
    \text{index of $x$} + d = \text{index of $y$}.
    \end{align*}
    This is because $y$ is the $d$-fold successor of $x$ if and only if the pair $(x,y)$ belongs to the relation produced in  Claim~\ref{claim:colour-distance-transduction}, and there is no $z$ such that both $(x,z)$ and $(z,y)$ are in this relation.   Once we have the $d$-fold successor and $(d+1)$-fold successor, we can get the $1$-fold successor, i.e.~the usual notion of successor. From the successor, we get the order by transitivity, thus finishing the proof of the lemma. 
\end{proof}

\subsubsection{Ensuring homogeneity}
\label{sec:ordered-ensuring-homo}
We now complete the proof of Lemma~\ref{lem:ordered-case}, by bridging the gap between the ordered approximation property in the conclusion of Lemma~\ref{lem:technical-condition-ordered} and the homogeneous variant in the assumption of Lemma~\ref{lem:use-ordered-approximation}. The proof is the same as in Section~\ref{sec:unordered-ensuring-homo}, except that the refinement step is more involved. (A linear preorder is seen as a partially ordered tree of height two, in the same way as for equivalence relations, except that the root is ordered.)
We only state the refinement step, which is given in Claim~\ref{claim:refine-ordered}, and which is the ordered version of Claim~\ref{claim:refine-unordered}. 

        \begin{claim}\label{claim:refine-ordered}
            Let $\beta$ be a transduction that outputs linear preorders, and has the unordered approximation property.
            There is a transduction $\gamma$ that refines $\beta$, has outputs of bounded height, and such that the local ordered transduction of $\gamma$ has the homogeneous ordered approximation property.  
        \end{claim}
        \begin{proof}
            An application of the Factorisation Forest Theorem of Imre Simon~\cite[Theorem 6.1]{simonFactorizationForestsFinite1990}.
        \end{proof}